\documentclass{article}

\usepackage{fullpage}
\usepackage[T1]{fontenc}
\usepackage{graphicx}
\usepackage{amsmath,amsfonts,amssymb}
\usepackage{hyperref}
\usepackage{xcolor}
\usepackage[english]{babel}
\usepackage[utf8]{inputenc}
\usepackage{framed}
\usepackage{float}
\usepackage{booktabs}
\usepackage[most]{tcolorbox}
\tcbset{left=1mm,right=1mm,top=1mm,bottom=1mm}
\usepackage{breakcites}

\newcommand{\mnote}[1]{}
\newcommand{\zxnote}[1]{}
\newcommand{\james}[1]{}
\newcommand{\adria}[1]{}
\newcommand{\todo}[1]{}

\def\ddefloop#1{\ifx\ddefloop#1\else\ddef{#1}\expandafter\ddefloop\fi}

\def\ddef#1{\expandafter\def\csname bb#1\endcsname{\ensuremath{\mathbb{#1}}}}
\ddefloop ABCDEFGHIJKLMNOPQRSTUVWXYZ\ddefloop

\def\ddef#1{\expandafter\def\csname c#1\endcsname{\ensuremath{\mathcal{#1}}}}
\ddefloop ABCDEFGHIJKLMNOPQRSTUVWXYZ\ddefloop

\def\ddef#1{\expandafter\def\csname v#1\endcsname{\ensuremath{\boldsymbol{#1}}}}
\ddefloop ABCDEFGHIJKLMNOPQRSTUVWXYZabcdefghijklmnopqrstuvwxyz\ddefloop

\def\ddef#1{\expandafter\def\csname v#1\endcsname{\ensuremath{\boldsymbol{\csname #1\endcsname}}}}
\ddefloop {alpha}{beta}{gamma}{delta}{epsilon}{varepsilon}{zeta}{eta}{theta}{vartheta}{iota}{kappa}{lambda}{mu}{nu}{xi}{pi}{varpi}{rho}{varrho}{sigma}{varsigma}{tau}{upsilon}{phi}{varphi}{chi}{psi}{omega}{Gamma}{Delta}{Theta}{Lambda}{Xi}{Pi}{Sigma}{varSigma}{Upsilon}{Phi}{Psi}{Omega}{ell}\ddefloop

\newcommand{\enc}{\mathsf{Enc}}

\newcommand{\dec}{\mathsf{Dec}}

\newcommand{\Agg}{\ensuremath{\mathsf{Agg}\hspace{2pt}}}
\newcommand{\Mode}{\ensuremath{\mathsf{Mode}}}
\newcommand{\State}{\ensuremath{\mathsf{State}}}
\newcommand{\Store}{\ensuremath{\mathsf{Store}}}
\newcommand{\Reveal}{\ensuremath{\mathsf{Reveal}}}

\newcommand{\ideal}{\ensuremath{\mathcal{I}}}
\newcommand{\server}{\ensuremath{\mathcal{S}}}

\newcommand{\ugets}{\stackrel{u}{\gets}}

\newcommand{\HintMLWE}{\mathsf{HintMLWE}}
\newcommand{\MLWE}{\mathsf{MLWE}}

\newif\ifeprint

\usepackage{amsthm}
\newtheorem{remark}{Remark}
\newtheorem{corollary}{Corollary}

\newtheorem{theorem}{Theorem}
\newtheorem{lemma}{Lemma}
\newtheorem{definition}{Definition}

\eprinttrue

\begin{document}
\ifeprint
\title{Secure Stateful Aggregation:\\ A Practical Protocol with Applications in Differentially-Private Federated Learning} %
\else
\title{Secure Stateful Aggregation\\ \small A Practical Protocol with Applications in Differentially-Private Federated Learning\\} %
\fi

\ifeprint
\author{
Marshall Ball \\ {\small New York University} \and
James Bell-Clark \\ {\small Google} \and
Adria Gascon\\ {\small Google} \and
Peter Kairouz\\ {\small Google} \and
Sewoong Oh\\ \small University of Washington \\ \small Google  \and
Zhiye Xie\\\small NYU Shanghai
}

\else
\author{Marshall Ball\inst{1}\orcidID{0000-1111-2222-3333} \and
James Bell-Clark\inst{3} \and
Adria Gascon\inst{3} \and
Peter Kairouz\inst{3}\orcidID{2222--3333-4444-5555} \and
Sewoong Oh\inst{3,4} \and
Zhiye Xie\inst{2}\orcidID{1111-2222-3333-4444}}

\authorrunning{F. Author et al.}
\institute{New York University \email{marshall@cs.nyu.edu}\and NYU Shanghai \email{zhiye.xie@nyu.edu}\and
Google \email{\{jhbell,adriag,kairouz,sewoongo\}@google.com} \and
University of Washington
}
\fi

\date{}
\maketitle

\begin{abstract}
Recent advances in differentially private federated learning (DPFL) algorithms have found that using \emph{correlated} noise across the rounds of federated learning (DP-FTRL) yields provably and empirically better accuracy than using independent noise (DP-SGD). While DP-SGD is well-suited to federated learning with a single untrusted central server using lightweight secure aggregation protocols, secure aggregation is not conducive to implementing modern DP-FTRL techniques without assuming a trusted central server. DP-FTRL based approaches have already seen widespread deployment in industry, albeit with a trusted central curator who provides and applies the correlated noise.

To realize a fully private, single untrusted server DP-FTRL federated learning protocol, we introduce \emph{secure stateful aggregation}: a simple append-only data structure that allows for the private storage of aggregate values and reading linear functions of the aggregates. Assuming Ring Learning with Errors, we provide a lightweight and scalable realization of this protocol for high-dimensional data in a new security/resource model, \emph{Federated MPC}: where a powerful persistent server interacts with weak, ephemeral clients. We observe that secure stateful aggregation suffices for realizing DP-FTRL-based private federated learning: improving DPFL utility guarantees over the state of the art while maintaining privacy with an untrusted central party. Our approach has minimal overhead relative to existing techniques which do not yield comparable utility. The secure stateful aggregation primitive and the federated MPC paradigm may be of interest for other practical applications.
\end{abstract}

\section{Introduction}
The widespread use of deep learning on user-generated data, that is often sensitive, has made privacy-preserving techniques increasingly important. One prominent framework that has emerged for conducting privacy-preserving machine learning is differentially-private federated learning (DPFL).

\paragraph{Differentially-private federated learning framework.}
While the exact details of differentially-private federated learning protocols may vary widely, many such systems~\cite{mcmahan2018learningdifferentiallyprivaterecurrent} follow a similar architecture:
\begin{itemize}
    \item A single central party plays the role of conductor for the learning process, grouping users into cohorts and facilitating communication. We refer to this persistent and powerful party as \emph{the server}, but it need not be localized onto a single device.
    \item Lightweight, ephemeral client devices (such as phones) are grouped into ``cohorts''. We refer to such these parties as \emph{clients}.
\end{itemize}

When a client's cohort's time comes, the device downloads the current model state and uses the on-device data to compute a local update (in gradient descent, this involves computing local gradients). These local updates are then aggregated across the entire cohort with some noise. The server then uses this noisy aggregate to update the model.

With this communication architecture and resource allocation in mind, we introduce a new practical model for designing secure multi-party computation at scale: \emph{secure federated multiparty computation (FMPC)}.
This model/paradigm can be seen as a hybrid of two emerging trends in secure multiparty computation: combining ephemeral, stateless participants (Fluid MPC~\cite{C:CGGJK21}, YOSO~\cite{C:GHKMNRY21}) with a powerful persistent central (untrusted) party (Gulliver MPC~\cite{C:ANOS24}\mnote{recent efficient arithmetic MPC approaches with ``king'' parties}).

\paragraph{Learning with independent noise (differentially-private stochastic gradient descent) and secure aggregation.}
Early successes in DPFL were due to a learning framework known as \emph{differentially-private stochastic gradient descent (DP-SGD)}~\cite{song2013stochastic,bassily2014private,abadi2016deep,DBLP:conf/mlsys/BonawitzEGHIIKK19}. DP-SGD learners perform stochastic gradient descent, but at each step, they clip gradients to bound their influence and add independent Gaussian noise to the gradient update to preserve privacy \emph{throughout} the learning process. These algorithms exhibited favorable privacy/utility tradeoffs and were easy to adapt to the federated learning architecture using lightweight protocols for \emph{secure aggregation}~\cite{DBLP:conf/ccs/BonawitzIKMMPRS17, bell2020secure, acorn, ma2023flamingo, opa}. 

A secure aggregation protocol enables a server to learn a sum of vectors and nothing else. Practical secure aggregation protocols are characterized by the ability to scale with very high dimensional data and massive numbers of participants (in contrast to generic secure multiparty computation): they should have very limited interaction (and a simple, sparse communication pattern), almost no communication overhead for high dimensional inputs (relative to the privacy-free baseline of sending inputs in the clear), low computational complexity, and robustness to client dropouts.

Given a secure aggregation protocol, one can turn DP-SGD into a federated learning protocol via the following: 
\begin{enumerate}
    \item The server distributes the (differentially-private) state of the model to the current client cohort.
    \item Clients locally compute a gradient update add some locally sampled noise\footnote{The local noise contributions need not be wide enough to provide privacy on their own: they need only provide privacy in aggregate. This is critical to yielding useful utility/accuracy guarantees at scale.} and securely aggregate the result across the entire cohort.
    \item The server then uses the aggregated gradient to update the model. Because the output of the aggregation is differentially-private (due to the local noise contributions), the new model also preserves differential privacy.
\end{enumerate}

A number of examples of secure aggregation protocols exist in the literature~\cite{DBLP:conf/ccs/BonawitzIKMMPRS17,USENIX:BGLLMY23,ma2023flamingo,li2023lerna} and indeed been proposed for deploying DP-SGD-based federated learning at scale.

The use of such techniques has been refined in a series of papers~\cite{chen2022fundamentalpricesecureaggregation,chen2022poissonbinomialmechanismsecure,kairouz2022distributeddiscretegaussianmechanism,agarwal2021skellammechanismdifferentiallyprivate} and this work has seen production deployment~\cite{DDPFLblog}.

Unfortunately, due to the less than optimal\footnote{DP-SGD utility can be improved using privacy amplification via sampling or shuffling, but these techniques are infeasible in the federated learning setting where data arrives in an arbitrary order.} privacy/utility tradeoff of DP-SGD, such federated learning procedures can yield underwhelming accuracy guarantees.~\cite{abadi2016deep,tramer2020differentially,kairouz2021practical,xu2023federated,choquette2024amplified,choquette2023correlated}

\paragraph{A new paradigm for private learning: correlated noise (differentially-private follow-the-regularized-leader).}
In recent years, a new paradigm for private learning has emerged. At a very high level, it has been observed both provably and empirically that by adding \emph{correlated noise} in the training steps, the utility can be dramatically improved while preserving the same level of privacy \emph{throughout} the training process. In particular, in such mechanisms, the noise added at different training steps is \emph{not} independently sampled.

This new family of algorithms, differentially-private follow-the-regularized-leader (DP-FTRL), is similar to DP-SGD but instead of adding independent noise $\eta_i$ to the gradient in round $i$, instead adds $\langle \vlambda^i,\veta\rangle$ where $\vlambda^i$ is a public vector (associated with round $i$) and $\veta$ is a vector of independent noise samples. 

While we shall give a simple example illustrating how such correlated noise can help improve privacy/utility trade-offs later (see~\ref{sec:running-sum}), for now, observe that the straightforward template for securely realizing DP-SGD via secure aggregation does not work here. Critical to that implementation was the fact that noise samples were independent in each learning step and, hence, amenable to local sampling by clients before aggregation. 

To apply the DP-FTRL paradigm, the underlying noise vector $\veta$ must persist throughout the entire learning process across the life of many cohorts, and it is not clear how to efficiently do this with secure aggregation alone. To date, no one has successfully realized private federating learning via such an approach without leveraging untenable trust assumptions on the central server.

We introduce a new primitive, \emph{secure stateful aggregation}, that enables a seamless realization of the DP-FTRL approach to differentially private federated learning. We provide a simple, scalable secure stateful aggregation protocol in the federated multiparty computation setting.

\subsection{Our Results}

We begin by introducing our conceptual contributions: the secure stateful aggregation functionality and the federated multiparty computation model. Then, we will sketch our stateful aggregation protocol and illustrate its applicability with a simple example: computing private partial sums.

\paragraph{Secure Stateful Aggregation.}

Secure stateful aggregation is a reactive functionality that can be thought of as a simple append-only data structure with two operations:
\begin{enumerate}
    \item $\Store$: Appends the sum of current inputs to the data structure state.
    \item $\Reveal$: Outputs a linear function of the current data structure state to the server.
\end{enumerate}

The actual functionality rolls these two operations together, but for the sake of clarity we provide this equivalent, albeit less efficient presentation. We refer the reader to Section~\ref{sec:func} for further details on the functionality and Figure~\ref{fig:ideal} in particular.

A secure realization of this functionality reveals nothing beyond its input/output behavior. The state of the data structure and any aggregated inputs will remain private, up to the linear functions that are revealed.

It is easy to see that this functionality is a mild generalization of the secure aggregation functionality. Moreover, this functionality allows for the secure aggregation of data across many cohorts. Moreover, this pared-down formulation allows for extremely efficient and nonetheless suffices for powerful applications in federated learning.

\paragraph{The $(\gamma,\beta)$-secure federated MPC paradigm.}
As mentioned above, we introduce a new paradigm for designing MPC protocols that aligns closely with many large scale distributed protocol deployments.

Concretely, a federated MPC (FMPC) protocol is broken into a sequence of rounds.
A powerful stateful server persists throughout the computation and is capable communicating with all participants. 
In each round, lightweight ephemeral clients are scheduled to arrive in a cohort. These clients have limited communication and computational capability and can only participate for at most a few rounds, sometimes just one.

Communication takes place on a bulletin board with a PKI: clients can send secure, private messages to clients in the successive cohort, but metadata about messages is visible to all participants (even if their contents are not). However, given the massive number of clients, each client can only send and receive a few short messages with other clients.  While larger messages to (and from) the server are possible, this communication should also be as close to the information-theoretic minimum as possible.

We assume all communication is effectively synchronous.

The adversary can corrupt an $\gamma$-fraction of any given cohort in addition to the server. Additionally, a $\beta$-fraction of the clients in any given cohort may drop out, failing to complete their roles in the protocol. In this work, we consider a semi-honest adversary making static corruptions and guaranteed output delivery.

The FMPC paradigm can be seen as a hybrid of two emerging trends in MPC: protocols with ephemeral participants (such as Fluid MPC~\cite{C:CGGJK21} and YOSO~\cite{C:GHKMNRY21}) and protocols with a strong central party and very weak clients (such as GMPC~\cite{C:ANOS24}). Federated MPC is comprised of a powerful persistent central party with a massive number of weak ephemeral clients. In contrast to Fluid MPC and YOSO, where minimal interaction is prized above all else, we assume a single persistent party. In contrast to the GMPC, we assume a much more reliable and transparent communication infrastructure for clients, and on the other hand that the clients are shortlived and unreliable.

It is our hope that this loosely-defined paradigm will help bring theory and practice closer together, at least in certain settings.

\paragraph{A lightweight protocol from RLWE.}
We provide a simple lightweight realization of stateful aggregation in the federated learning setting. Our protocol scales well with high-dimensional data and massive client cohorts. %

The key ingredient is a (high-rate) linearly homomorphic secret key encryption scheme that also admits a kind of key homomorphism (enabling distributed encryption and decryption), which we instantiate via the Ring Learning with Errors (RLWE) assumption~\cite{lyubashevsky2010ideal}.

The high level idea is very straightforward. Throughout, a persistent global secret key is reshared from cohort to cohort using additive secret sharing (reminiscent of DC-nets~\cite{JoC:Chaum88}).
\begin{enumerate}
    \item To $\Store$ an aggregate: clients encrypt their private inputs using their share of the secret key and the server aggregates the ciphertexts to produce an encryption of the aggregate under the global key and appends the result to its state. This is possible due to the key homomorphism of the scheme.
    \item To $\Reveal$ a linear function of the secret state: the server homomorphically evaluates the linear function over ciphertexts it is holding. Then the clients use their secret shares to run a distributed decryption of the resulting ciphertext. 
\end{enumerate}
Ensuring that this does not inadvertently compromise semantic security of the state (up to the linear function output) is slightly delicate. An elementary committee-based approach is proposed to handle client dropouts. We refer readers to Section~\ref{sec:protocol} for more details.

In our protocol, communication to the server approaches the size of the aggregated elements. %
Clients send roughly $\kappa^2$ messages of length $\kappa$ to other clients, where $2^{-\kappa}$ is the desired security level. \mnote{check this is consistent description} The server is completely silent and hence the protocol is actually secure against a malicious server by default. (We only guarantee security against semi-honest clients.)  Precise benchmarking can be found in Section~\ref{sec:benchmark}.

Computational demands on clients are comparable to the communication costs.

\paragraph{Application 1: Releasing private partial sums (introduction to correlated-noise mechanisms).}~\label{sec:running-sum}
Let $x_i$ denote the value of the aggregated inputs in cohort $i$. Consider the task of releasing a differentially-private running sum: 
\[x_1, x_1+x_2, x_1+x_2+x_3,\ldots, \sum_{j=1}^i x_j,\ldots, \sum_{j=1}^N x_j \]

A naive mechanism for doing this (assuming a trusted central curator) is to simply add independent noise (be it Gaussian or Laplacian or otherwise) to each output above, treating each output as a single count mechanism. Then by composition, we can argue that the whole release is differentially private. The downside of this approach is that the input $x_1$ appears in every single output, which means (according to the DP composition theorem) that privacy will degrade by a factor of $N$. This means that to achieve $\epsilon$-privacy we need to use $\epsilon/N$-private noise on each component, which potentially drowns out the partial sums entirely. Note that it is unclear how to realize even this naive mechanism using secure aggregation, but it can be trivially realized using stateful secure aggregation: in each round, clients can store their aggregated inputs and noise samples separately and reveal their aggregatation of all inputs.

A clever mechanism introduced by Dwork et al.~\cite{DNPR10} (referred to as Tree Aggregation in the literature) proposes to instead think about a complete binary tree with $x_1,x_2,...,x_N$ at its leaves. Label each internal node with the sum of the nodes beneath it (so the root is labeled with $\sum_{j=1}^N x_j$). Now, simply add independent noise to the label of each node in the tree and release all the noisy labels. To estimate the $i$the partial sum, one can compute a linear function of the noisy labels.\footnote{A straightforward approach simply sums the roots of the any left children on path to $i$th leaf (and the $i$th leaf), but better approaches are possible~\cite{honaker2015efficient,kairouz2021practical}.} Because any input appears in at most $\log_2 N$ labels, to achieve $\epsilon$-privacy one need only use $\epsilon/\log N$ noise for each internal node.

Now, it is easy to see how to create a mechanism that releases this entire noisy tree using secure stateful aggregation. However, recall that the ultimate $i$ partial sum is a linear function, $\vlambda^i$, of the noisy internal labels. This estimator is constructed so that the output is equal to the partial sum plus a linear function of all the independent noise in the tree $\veta$ : 
\[ \sum_{j=1}^i x_j + \langle \vlambda^i,\veta\rangle.\]

From this perspective, we can imagine a new mechanism that samples persistent independent noise $\veta$ (possibly in an offline phase, although it is possible to construct this mechanism so the noise can be generated in an online manner, one new independent component per partial sum), then at step $i$ directly releases the value above. To realize this new mechanism via secure stateful aggregation is again quite straightforward: clients aggregate their inputs and separately aggregate a fresh noise sample. Then, the expression above can be released by computing the appropriate linear function of the state: releasing the value described above.\footnote{The advantage of this alternate approach (over releasing the entire tree), is that the partial sum estimator can be released in one shot as soon as the data is available. Releasing the whole noisy tree, while also possible and private, not only requires many releases at certain times but also requires more releases (leading to a higher cost overall).}

We can rephrase this last mechanism in linear algebraic terms. Let $\vC$ be the matrix that maps to cohort input aggregates $\vx=(x_1,\ldots,x_N)$ to tree labels as described above. Then let $\vB$ denote the matrix such that the $i$th row is $\vlambda^i$, the linear function that produces the $i$th partial sum estimator. Then, in this notation our mechanism will output 
\[\vB(\vC\vx +\veta) = \vy,\]
where $\vy=(y_1,\ldots,y_n)$ and $y_i = \sum_{j=1}^i x_j + \langle \vlambda^i,\veta\rangle$.

\paragraph{Application 2: Differentially-private federated learning via DP-FTRL.}

Kairouz et al.~\cite{kairouz2021practical} developed a new approach for training models with differential privacy using batch gradients. At a very high level their idea was that for convex optimization, one can bound regret by looking at a linear function of the loss of the model at time $i$. This means that the next model step can be computed using a linear function of the batch gradients seen thus far, essentially a partial sum of these gradients. Thus training a model privately effectively reduces to building a mechanism for releasing iterative partial sums as in Application 1 above! 

Thus we can effectively use the same mechanism, albeit with a different choice of matrices $\vC$ and $\vB$, as the one described above to perform state-of-the-art differentially private learning at scale. In particular, we will choose $\vB$ and $\vC$ such that $\vA=\vB\vC$ where $\vA$ is the linear function that maps the sequence of batch gradient described above. A ``good'' factorization of $\vA = \vB\vC$ yielding nearly optimal privacy/utility tradeoff for this paradigm can be found using semi-definite programming.~\cite{kairouz2021practical,choquette2024amplified} The stateful secure aggregation then need only output 
\[\vB(\vC+\veta)= \vA\vx+\vB\veta.\]
At the $i$th training step, this is simply $\vA^i\vx+\vB^i\veta$ where $\vA^i,\vB^i$ are the $i$th rows of $\vA$ and $\vB$, respectively.

This new approach has been shown to yield dramatic improvements over other methods such as DP-SGD approaches, both theoretically~\cite{kairouz2021practical,choquette2023correlated} and empirically~\cite{kairouz2021practical,choquette2024amplified}. Moreover this is already being deployed at scale for private next word prediction in the Google Keyboard~\cite{xu2023federated}, albeit with a trusted server.\footnote{User gradients in this process are aggregated securely, yielding some privacy guarantees. However, while the model outputs preserved differentially-privacy to external parties and users, Google's internal view was not differentially-private. By using our secure stateful aggregation, the view of \emph{all} parties remains differentially-private.}

\subsection{Limitations and Future Directions}

We suspect that the protocols listed here are maliciously secure, assuming the presence of a PKI and a means of choosing which parties should take part in each round. However, in practice those would be very big assumptions in the presence of a malicious server so we have not prioritized showing this formally. Demonstrating malicious security and finding ways to integrate with realistic ways to resist Sybil attacks would be a useful contirbution.

Our implementation of MF-DP-FTRL requires the matrix $\vC$ to be banded. In the central model concurrent work has shown that is is possible to work with Buffered Toeplitz matrices~\cite{mcmahan2024hasslefreealgorithmprivatelearning} giving slightly better results. That approach however doesn't interact well with discretization, so we cannot straight forwardly extend to it. Finding a way to make that work is a possible future direction.

Of course it isn't obvious that the communication and computation couldn't be reduced by practically meaningful constant factors by some other approach. This could also be a aim of future work.

\section{Stateful Aggregation}\label{sec:func}

The system consists of a server and a sequence of $r$ cohorts of clients. The $i$th cohort $C_i$ is available to make a submission at time step $i$. It must also know the public keys for $C_{i+1}$ and $C_{i+2}$ at that time, thus $C_i$ must have been chosen and provided their public keys by time $i-2$.

In each step, multiple aggregations can be conducted. Each of these aggregations computes a linear combination of inputs from the clients that step and values stored in the protocol's memory in previous steps. The result of each aggregation can either be stored in memory or revealed to the server.
For notational clarity we will assume that exactly one aggregation is conducted per step.

An aggregation instruction has three arguments: one indicating whether the result should be revealed, taking values in $\{\Store,\Reveal\}$; a rule explaining what value each client should provide for aggregation, taking values in a set $\mathcal{I}$; and the weights to be applied to stored values for inclusion in the aggregation which form a value in $\mathbb{F}_q$ for some prime $q$. An aggregation in the $i$th round can thus be written as either $\Agg(\Store,I_{i},\{\lambda_{i,k}\}_{k<i})$ or $\Agg(\Reveal,I_{i},\{\lambda_{i,k}\}_{k\leq i})$. The input rule $I_{i}$ can be any object that the clients know how to use to derive their inputs to the specific aggregation, we will write $\vx_{i,j}$ for the input that the $j$th client in $C_i$ derives from $I_{i}$. A program $P$ for our system consists of a sequence of aggregations, one per round. The $i$th entry in the sequence being the  aggregation for the $i$th cohort to perform. We say a program of that form is valid if 
the dependency graph defined by input and outputs of its instructions is acyclic.

We make two simplifying assumptions, neither is a hard restriction but they will keep things simpler and hold for the applications we have in mind. Firstly, we assume that the program is known in advance, although it would be possible to decide the $i$th entry of the sequence after the $(i-1)$th cohort have spoken. In the malicious server case the $i$th cohort would have to check that the $i$th entry of the sequence was in some sense legitimate or the adversary could ask for the submission of secrets it shouldn't learn, but this would be doable for many applications. Secondly, as mentioned above, each cohort conducts at most one $\Store$ and one $\Reveal$ aggregation. Let $v_i$ be the value stored or revealed in the $i$th round. We define the following ideal functionality.

\begin{figure}\small
\begin{tcolorbox}[enhanced, title=Stateful Aggregation Functionality $\ideal$]
\begin{enumerate}
\item[] \hspace{-0.4cm} \textbf{Public Parameters:} 
\begin{itemize}
\item A program to be executed $P := \big\langle t_1, \ldots, t_r  \big\rangle$.
Each instruction $t_{i}$ is of the form  $(\Mode_i, I_{i}, \{\lambda_{i,k}\}_{k<i})$, 
with $\Mode\in \{\Store, \Reveal\}$. 
\item Vector length $\ell$, input domain $\mathbb{F}^\ell$.
\end{itemize}
\item[] \hspace{-0.4cm} \textbf{Parties:} 
\begin{itemize}
    \item A Server $\server$.
    \item A sequence of (not necessarily disjoint) cohorts $C_1, \ldots, C_r$ of clients. Each cohort contains $n$ clients. We denote by $C_{i,j}$ the $j$th client in the $i$th cohort. Each client holds a private database $D_{i,j}$. %
    \item A trusted party $\ideal$ holding state $\State := \langle (v_i) \in \mathbb{F}^{\ell}\rangle_{i\in [\leq r]}$, initially $\langle \rangle$ (empty).
\end{itemize}
\item[] \hspace{-0.4cm} \textbf{Private Parameters:} Client $C_{i,j}$ may hold private input $\vx_{i,j}:= I_i(D_{i,j})$ for round $i$.
\item[] \hspace{-0.4cm} \textbf{Functionality:} 
\item[] For each round $i\in [r]$:
\begin{itemize}
\item Let $(\Mode_i,I_i,\{\lambda_{i,k}\}_{k<i}) = t_i$.
\item Each client $C_{i,j}$ sends $\vx_{i,j}$ to $\ideal$.
\item $\ideal$ computes $v_i := \sum_{j\in [n]} \vx_{i,j} + \sum_{k<i} \lambda_{i,k}v_k$.
\item $\ideal$ updates $\State := \State~||~ v_i$.
\item If $\Mode_{i-1} = \Reveal$ then $\ideal$ sends $v_{i-1}$ it to $\server$.
\end{itemize}
\end{enumerate}
\end{tcolorbox}
\caption{The Stateful Aggregation Functionality}
\label{fig:ideal}
\end{figure}

\paragraph{Long running aggregation.} As a basic example, consider the case where the server just wants to compute the sum of the private input from clients across different round. Let $\texttt{input}$ be a rule that just returns the client's input to the server. In this case each cohort except the last one will run $\Agg(\Store,\texttt{input},\{\mu_{i,k} := 0\}_{k<i})$ . The final cohort will run only $\Agg(\Reveal,\texttt{input},\{\nu_{i,k} := 1\}_{k<i})$.

\section{Applications of Stateful Aggregation in Distributed Differential Privacy}

In this section, we describe a couple of programs for releasing prefix sums with distributed differential privacy using our system. We have a sequence of cohorts, each client in these cohorts has a vector as input, let the sum of inputs in cohort $i$ be $x_i$. The server should after each cohort receive an estimate of the sum of all inputs from all clients who have submitted so far, which we call $S_i$, i.e. after cohort $i$ it should receive an estimate of $S_i := \sum_{j\leq i} x_j$. This is useful in federated learning, where the inputs are user contributed updates to a model and the current parameters are the sum of the initial parameters and all inputs so far.

In both cases the noise will be added by the clients making the submissions and for privacy we will assume that at most a fraction $\gamma$ of them are corrupt (which is also an assumption for the security of the protocol so this is no extra assumption).

Following a standard trick, when we want a Gaussian random variable with variance $\alpha^2$ for differential privacy, each client can provide Gaussian noise with variance $\alpha^2/n(1-\gamma)$. The sum of honest contributions will then have the correct variance. For privacy purposes any extra noise added can be considered post-processing. The effect of this approach is to inflate the variance of the added noise by a factor of $1/(1-\gamma)$, this will be small compared to our other gains.

Throughout we assume that each client's input has a bounded sensitivity and that in order to achieve the required DP epsilon locally the required variance of Gaussian noise (for each coordinate) would be $(1-\gamma)\sigma^2$ to each entry. The choice of $\sigma$ will in practice depend on sensitivity, $\gamma$ and $\epsilon$, but the total noise in all of the following methods will depend on those parameters only through $\sigma$, thus this will simplify presentation.

\subsection{Baselines}

If we don't require differential privacy we can have a secure aggregation for each cohort providing the sum $x_i$ of that cohort's inputs. This provides the server with the difference between their output for this round and their output for the previous round, which they would learn anyway and from which they can compute their current output $\sum_{j\leq i} x_i$ as $x_i$ plus their previous output $\sum_{j<i}x_i$.

This extends naturally to the baseline differentially private idea of using a separate secure aggregation with each cohort to produce $x_i+z_i$ where the $z_i$ is a Gaussian noise with variance $\sigma^2$. However, the prefix sum that the server outputs in the $j$th round is $\sum_{i\leq j} x_i+z_i$ which has noise $\sum_{i\leq j} z_i$ which has variance proportional to $j\sigma^2$, we can do better than this.

\subsection{Prefix Tree Aggregation}

Dwork et. al.~\cite{DNPR10} suggested the following DP mechanism that has since come to be known as Tree Aggregation. Suppose we have $2^h$ cohorts for some integer $h$. Assign the cohort's inputs $x_i$ to the leaves $l_i$ of the tree from left to right. Assign Gaussian noise samples $z_i$ to each of the left child nodes in the tree and to the root, such that $z_i$'s node $n_i$ has $l_i$ as its rightmost descendant (about half the time $n_i=l_i$). Let $\upsilon(i)$ be the index of the leftmost descendant of $n_i$. At time step $i$ the curator calculates $r_i = z_i+\sum_{j=\upsilon(i)}^i x_j$, that is the true sum of the leaves descended from $n_i$ masked by the noise at $n_i$. It then adds the output from step $\upsilon(i)-1$ (or nothing if $\upsilon(i)=1$) and outputs the result.

As the full sequence of outputs is a post processing of the $r_i$ it suffices to show that the $r_i$ are cumulatively DP. As each input $x_i$ has at most $h$ ancestors amongst the $n_i$, it is included in at most $h$ of the $r_i$. By advanced composition it is thus sufficient for each of the noises to have variance $O(k\sigma^2)$. Each output is the sum of at most $h$ of the $r_i$ and thus the variance of the noise on each output is $O(h^2\sigma^2)$ i.e. $O(\log_2(r)^2\sigma^2)$. The constants resulting from this protocol can be optimized at the cost of extra computation was described by Honaker~\cite{honaker2015efficient}, we will not bother implementing a version of that as it is more complicated (so not interesting for exposition) and superfluous given MF-DP-FTRL.

We now provide a program $P=\langle t_1,...,t_{2^{h+1}}\rangle$ for our functionality that will implement the above. In this program the odd numbered cohorts will not provide data, only Gaussian noise to be used for DP and the even numbered cohorts will be the cohorts providing the noise, thus $x_i$ will be uploaded in round $2i$.

We remark that there is no reason why the same physical devices couldn't play the role of cohorts $2i-1$ and $2i$, further this would avoid the cost of transferring the key for that change, that would further allow the work for the two rounds to be done in parallel. This is probably how the protocol would be run in practice but we haven't explained the details here to keep the exposition and interface simple.

Let $G_{\sigma^2}$ be a function that generates and returns a vector in $\mathbb{F}^l$ of discrete Gaussians with variance $\sigma^2/n$. In round $2i-1$ we will have the cohort generate and store the noise $z_i$ that is we take $t_{2i-1}=(\Store,G_{\sigma^2},{0}_k)$. In round $2i$ we will have the server learn the $x_i$ plus the noise that we want applied to $S_i$ minus the noise that was applied to $S_{i-1}$. The server can then add this to its previous output to get the output for round $i$.

Let $2^{h_i}$ be the largest power of two dividing $i$. The difference of the noise for $S_i$ and $S_{i-1}$ is given by $z_i-\sum_{d=0}^{h_i-1}z_{i-2^d}$. We let $I$ map a client's data to the input we want them to provide to the aggregation and define $\lambda_{2i,k}=-1$ if $k\in \{2i-2^d-1|d\in\{1,...,h_i\}\}$ and $\lambda_{2i,k}=0$ otherwise. Then the even indexed instructions in our program are given by $t_{2i}=(\Reveal,I_{2i},\{\lambda_{2i,k}\}_{k<2i})$.

\subsection{MF-DP-FTRL}

The state of the art in central model DP federated learning is given by the matrix factorization approach~\cite{choquette2024amplified}. We now describe an outline of this procedure and the optimizations provided in that paper.

Define $\vA$ to be the lower triangular matrix with all entries (on and below the diagonal) equal to one. Note that if $\vx$ is a vector (of vectors) with the $i$th cohort's contribution in the $i$th place then the task we are aiming for is to estimate $\vA\vx$ in a streaming fashion.

The matrix factorization in the name is of $\vA$ into two components $\vA=\vB\vC$. The calculated result will be given by $\vB(\vC\vX+\veta)$ for some Gaussian noise $\veta$. To prove that this is DP it is enough to show that $\vC\vX+\veta$ is DP, which is done by requiring $\vC$ to have Frobenius norm at most one and setting the variance of $\veta$ to be the same as would make $\vX+\veta$ DP. Thus the factorization is usually chosen to minimise $\vB\veta$ subject to the bound on the norm of $\vC$. This optimization is then done numerically and results in substantial practical improvements over the Tree Aggregation idea above.

In practice it is also important that the necessary matrix multiplications can be calculated efficiently in an online fashion. Efficiently here largely means that the server doesn't want to have to store $\Omega(r)$ vectors in memory at any point. To achieve this it is recommended to choose $\vC$ to be banded with band width $b$. The result can then be calculated as $\vA\vC^{-1}(\vC\vX+\veta)$ using online algorithms for multiplying by a known banded matrix or its inverse that each require storing only $b$ vectors at any point (these are given by Algorithms~8~and~9 in~\cite{choquette2024amplified}). Adding the restriction that $\vC$ is bounded is shown numerically to lead to little loss in utility and so this restriction is recommended.

In our protocol we will also use the fact that $\vC$ is banded to get an efficient protocol in runtime and storage in much the same way. We will add one more restriction on $\vC$ which is that we require it to be discrete. This is a minimal change because we can discretize at any fixed level of precision. Thus we propose optimizing $\vC$ over the reals as in the central model and then rounding each entry in $\vC$ to discrete values. This may increase the Frobenius norm of $\vC$ to slightly more than $1$, if this happens then rounding down some of the entries that were barely rounded up should bring it back down without significantly damaging the fidelity of the approximation. If all rounding is toward $0$ then the Frobenius norm will not increase and the fidelity will still be good for sufficiently fine discretizations.

We note that the multiplication by $\vA\vC^{-1}$ can be considered post processing and so can be done in the clear. It is enough to implement the online processing of $\vC\vX+\veta$ using our system.

Again to have $i$ cohorts provide inputs we will run a program $P$ with $2i$ instructions. The same possibilities for combining these rounds apply as in the Tree Aggregation case. In the Tree Aggregation case where the noise was stored and then applied to each input as it was revealed. In this case we will store the inputs in the odd numbered rounds and then in each even numbered round reveal a new instance of noise with the appropriate linear combination of the inputs on top. The odd indexed instructions are thus $t_{2i-1}=(\Store,I,\{0\})$ and the even ones are $t_{2i}=(\Reveal,G_{\sigma^2},\{\lambda_{2i,2k-1} = C_{i,k}\}_{k\leq i})$.

The outputs from this Program can then be scaled back from fixed point to floating point encodings and then online multiplied by $\vB=\vA\vC^{-1}$ as in the central model.

\section{Realizing Secure Stateful Aggregation}\label{sec:protocol}
For clarity, we begin by describing a protocol for securely realizing this functionality in the fully-synchronous (or no client dropouts) semi-honest setting, and provide some intuition for its security. (A formal security proof for this protocol can be found in Appendix~\ref{sec:no_dropout-proof}) This setting captures the key ideas in realizing stateful secure aggregation. In Section~\ref{sec:dropout}, we describe how to augment this basic protocol to achieve resilience to client dropouts.

Before continuing, let us recall the stateful secure aggregation functionality. A stateful secure aggregation program consists of a sequence of instructions $t_i=(\Mode_i,I_i,\vlambda^i)$ and maintains an append-only data structure whose state at time $i$ we denote $\vv$. 

At time $i$ when instruction $t_i$ is executed, all clients currently present in cohort $C_i$ ($|C_i| = n$) submit their inputs $\vx_{i,1},\dots,\vx_{i,n}$ to the server. Then the sum of these inputs with a linear function, $\vlambda^i \in \bbZ_q^*$ for some prime $q$, of the prior state is appended to the state: $\vv \gets \vv||v_i$ where $v_i = \sum_{j\in [n]} x_{i,j} +  \langle\vlambda^i,\vv\rangle$. 

If $\Mode_i=\Store$, then this is all that happens. No output is produced. Otherwise, if $\Mode_i=\Reveal$, then value just appended to the state, $v_i$, will be released in the next time step.

We will show how to securely implement this functionality in the presence of a semi-honest adversary who can corrupt at most an $\gamma$-fraction of any cohort and the central server. We then show how to augment this protocol to achieve robust correctness (and security) guarantees in the presence of a fail-stop adversary who can force corrupt clients to drop out.

\subsection{Preliminaries: Linearly Homomorphic Encryption with Distributed Encryption/Decryption via (R)LWE}
The key ingredient in our scheme is a simple symmetric key encryption scheme based on learning with errors (LWE) assumptions that admits both key and message homomorphism. In particular, given a key $\vA,\vs$ where $\vA$ can be public, a message $\vx$ is encrypted as
\[
\enc_{\vA,\vs}(\vx) \to \vA\vs + T\ve + \vx
\]
where $e$ is sample from some appropriate small noise distribution and $T$ is an appropriately chosen scalar.

To decrypt a ciphertext $\vc=\vA\vs + T\ve + \vx$, one subtracts $\vA\vs$ and removes the noise:
\[
    \dec_{\vA,\vs}(\vc) = \vc - \vA\vs \mod T = (\vA\vs + T\ve + \vx) - \vA\vs \mod T = T\ve +\vx \mod T = \vx
\]

The first key property that we will rely on is linear message homomorphism: given encryptions of $\vx$ and $\vy$ under the same secret key (but possibly different public keys), one can produce an encryption of $a\cdot\vx+b\cdot\vy$ (albeit with respect to a different public key).
\[
    a\cdot(\vA\vs + T\ve + \vx) + b\cdot(\vB\vs+T\vf +\vy) = (a\vA+b\vB)\vs + T(a\ve+b\vf) +a\vx+b\vf
\]
So long as the coefficients $a$ and $b$ are appropriately bounded (and hence the noise $a\ve+b\vf$ and message $a\vx+b\vy$ are not too large), this can be correctly decrypted.

The second key property we rely on is key homomorphism, which enables a form of distributed encryption and decryption. In what follows, imagine Alice and Bob are holding $\vs_1$ and $\vs_2$ additive shares of the secret key $\vs$ such that $\vs_1 +\vs_2 = \vs$.

To compute a distributed encryption of the sum of their inputs ($\vx_1,\vx_2$ respectively), Alice can send $\vA\vs_1+\vx_1+T\ve_1$ and Bob can send $\vA\vs_2+\vx_2+T\ve_2$. The server can sum the result to get an encryption of $\vx_1+\vx_2:$ $\vc=\vA\vs+(\vx_1+\vx_2) + T(\ve_1+\ve_2)$. So long as $\ve_1+\ve_2$ is small (which is the case if $\ve_1$ and $\ve_2$ are small, $\vc$ can later be correctly decrypted.

Now to see how key homomorphism enables distributed decryption, imagine the server is holding a cipher text $\vc=\vA\vs+T\ve+\vx$. Now, Alice and Bob can simply compute and send $\vA\vs_1$ and $\vA\vs_2$ respectively. This enables the server to recover $\vx$
\[
    (\vA\vs + T\ve + \vx) - \vA\vs_1 - \vA\vs_2 \equiv_T \vx
\]
Unfortunately, this also allows the server to recover $\ve$, and in turn $\vs$. While this may be ok in a one-time scenario, we will require a distributed decryption that only reveals ``safe'' leakage on $\ve$ (or following the terminology of Lee et al.~\cite{ePrint:LeeKKSSC18,ACISP:CheonKKLSS21,USENIX:BGLLMY23}: hints about $\ve$) that won't compromise $\vs$. Following Bell et al.~\cite{USENIX:BGLLMY23}, we note that semantic security on correlated ciphertexts can be preserved without impinging upon correctness if Alice and Bob add some  noise to their messages (effectively sending encryption of 0 using their private keys):
\[
    (\vA\vs + T\ve + \vx) - (\vA\vs_1 - T\ve_1) - (\vA\vs_2 - T\ve_2) = \vx + T(\ve+\ve_1+\ve_2) \equiv_T \vx
\]
The server now can learn $\ve+\ve_1+\ve_2$, but as shown in~\cite{ePrint:LeeKKSSC18,USENIX:BGLLMY23} this preserves semantic security ($\vu$ is uniformly random below):
\[
    (\vA\vs+T\ve,\ve+\ve_1) \approx (\vu, \ve+\ve_1)
\]

Before continuing, we note that these properties are satisfied by other encryption schemes.\footnote{For example, our framework can be instantiated with ElGamal (provided one sufficiently constrains the message space---which suffices for our applications). However, this significantly degrades the complexity of communication relative to the RLWE-based approach.} However, in our study, the scheme below instantiated with Ring Learning with Errors yielded the best practical parameters.

\subsection{Fully-Synchronous Semi-Honest Protocol (No dropouts)}
We begin by showing how to securely realize the stateful secure aggregation functionality in the absence of client dropouts. We will informally describe this protocol and give intuition for its security. The formal description of the protocol can be found in Figure~\ref{fig:server_no_dropout} and Figure~\ref{fig:client_no_dropout}. A formal security proof can be found in Appendix~\ref{sec:no_dropout-proof}.

The high-level idea of our secure stateful aggregation protocol is relatively straightforward. 

\paragraph{A persistent secret key.} At the outset, clients in the first cohort, $C_1$, locally and independently sample uniformly random secret keys, $\vs_{1,1},\ldots,\vs_{1,n}$. This implicitly defines a global secret key $\vs = \sum_{j=1}^{n} \vs_{1,j}$. Throughout the protocol, we will maintain the invariant that the clients of any particular cohort are holding an additive secret sharing of $\vs$. 

To do this, we use a simple trick reminiscent of Chaum's dining cryptographers~\cite{JoC:Chaum88}. If the $j$th client in cohort $C_i$ is holding a share $\vs_{i,j}$, that client simply additively shares $\vs_{i,j}$ into $\vs^1_{i,j},\ldots,\vs^d_{i,j}$ such that $\vs^1_{i,j},\ldots,\vs^d_{i,j}$ are uniform conditioned on $\sum_{k=1}^d \vs_{i,j}^{k}=\vs_{i,j}$. Then $C_{i,j}$ sends those shares to $d$ randomly chosen clients in the next cohort. Provided that every honest client
sends a message to some other honest client and receives at least one message from some other honest client, $C_{i,j}$'s share $\vs_{i,j}$ remains perfectly hidden. Clients of the next cohort simply sum up the shares they receive to produce their own share. In particular, if the $k$th client in cohort $i+1$, receives $\bar{\vs}^1,\ldots,\bar{\vs}^{d'}$, then its share of the secret key will be $\vs_{i+1,k}=\sum_{j=1}^{d'} \bar{\vs}^j$. 

It is easy to verify that the invariant is maintained and so long as $d$ is sufficiently large, no honest party's secret key share will be compromised.

We will use this persistent secret key to encrypt the state of the protocol, $\vv=(\vv_i)_{i \leq r}$. In particular, the server will hold an ever growing sequence of ciphertexts $\vv=(\hat{\vv_i})_{i \leq r}$ where the $i$th ciphertext is an encryption of $\vv_i$ relative to a public matrix $\vA_i$ and the global secret key $\vs$. 

\paragraph{Writing to the secret state.} Having established how to maintain private random keys that sum to the same key $\vs$ at any given time, we next describe a simple mechanism for updating an encrypted state held by the server. Each client $j$ in cohort $C_i$ simply uses their secret key share $\vs_{i,j}$ to encrypt their input $\vx_{i,j}$. The clients then send these ciphertexts to the server. The server, holding an encryption of the old state $\vv$, uses the linearly holomorphic property of the encryption scheme to compute an encryption of $\langle \vlambda^i,\vv\rangle$. The server then simply sums the resulting correlated ciphertext with all ciphertext received from cohort $C_i$. The result, an encryption of $\langle \vlambda^i,\vv\rangle + \sum_{j=1}^{n} \vx_{i,j}$, is then appended to its encrypted state.

\paragraph{Revealing parts of the secret state.}
If  $\Mode_{i-1}=\Reveal$, then we will use the distributed decryption property to open the last part of the state. Namely, clients in cohort $C_i$ send messages for distributed decryption of the last part of the state. The server then uses these messages to reveal that part of the state.

\paragraph{Security intuition.} Arguing security amounts to proving that these distributed decryptions of homomorphically evaluated ciphertexts preserve semantic security of the underlying ciphertexts sent by clients when writing to the state, up to some linear constraints, even when this is done repeatedly.

We do this via a hybrid argument over the individual messages, but the step is in arguing that opening linear functions of a sequence of ciphertexts is indistinguishable from uniform (up to the outputs of the linear functions). We argue this by reducing to a generalization of $\mathsf{HintMLWE}$ introduced by Kim et al.~\cite{C:KLSS22} (Def.~\ref{def:HintMLWE}\footnote{Our definition is slightly different than that of Kim et al.~\cite{C:KLSS22}: we consider a variant with uniformly random secrets (as opposed to Gaussian) but do not reveal leakage on the secret.} ) that allows for multiple leaks or hints on the noise:
\[
    (\vA\vs + T\ve,\ve+\vf^1,\ldots,\vf^\ell)\approx (\vu,\ve+\vf^1,\ldots,\ve+\vf^\ell).
\]
By a direct reduction (see Corollary~\ref{cor:linear-leakage} for precise statement and details), this implies that for $\vlambda^1\ldots\vlambda^\ell$ there exist some noise distributions ${\vf}^1,\ldots,{\vf}^\ell$ such that
\begin{align*}
    \bigg(\underbrace{\vA_1,\ldots,\vA_r}_\text{public matrices}, \underbrace{\vA_1\vs+T\ve_1 + \vx_1,\ldots,\vA_r\vs+T\ve_r+\vx_r}_\text{server's encrypted state (aggregated $\Store$ messages)},  \underbrace{\sum_{i=1}^r\vlambda^{1}_i(T{\vf}^1-\vA_i\vs),\ldots,\sum_{i=1}^r\vlambda^{\ell}_i(T{\vf}^\ell-\vA_i\vs)}_\text{client's distributed decryptions (aggregated $\Reveal$ messages)}\bigg)\\
    \qquad \approx \bigg(\underbrace{\vA_1,\ldots,\vA_r}_\text{public matrices}, \underbrace{\vu_1,\ldots,\vu_r}_{\text{uniform}}, \underbrace{\sum_{i=1}^r\vlambda^{1}_i(T{\vf}_i^1+T{\ve}_1-\vu_i)+ \langle \vlambda^1,\vx\rangle,\ldots,\sum_{i=1}^r\vlambda^{\ell}_i(T{\vf}_i^{\ell}+T{\ve}_1-\vu_i) + \langle \vlambda^\ell,\vx\rangle}_\text{simulated decryption aggregation}\bigg)
\end{align*}
The first $r$ components are the public matrices in both distributions. The second $r$ component can be thought of as the server's encrypted state (the sum of the $\Store$ messages). The third $\ell$ components are distributed decryptions (the sum of the $\Reveal$ messages).

Critically note that if one considers the function
\[
    \phi^i: (\vA_1,\ldots,\vA_r,\vb_1,\ldots,\vb_r,\vc_1,\ldots,\vc_{\ell}) \mapsto \vc_i + \sum_{j=1}^r \vlambda^i_j\vb_j,
\]
then $\phi^i$ applied to the either distribution yields
\[ \sum_{j=1}^r \vlambda^i_j\vx_j + T\left(\sum_{j=1}^r\vlambda^i_j(\vf^i_j+\ve_j)\right).\]
Thus, provided $\vlambda^i$ is appropriately bounded, the server can correctly recover $\sum_{j=1}^r \vlambda^i_j\vx_j$.

On the other hand, the indistinguishability of these two distributions means that (provided the client's messages are securely aggregated) nothing is leaked to the server about the state beyond precisely $\langle \vlambda^1,\vx\rangle,\ldots,\langle \vlambda^\ell,\vx\rangle$. From there is simply a matter of arguing that the aggregation is secure using a hybrid argument (similar to Bell at al.~\cite{USENIX:BGLLMY23}).

This is summarized in the following theorem:
\begin{theorem}
    Assuming that a semi-honest PPT adversary corrupts at most an $\gamma$-fraction of any user cohort, in addition to the server, 
    The protocol given in Figures~\ref{fig:server_no_dropout},~\ref{fig:client_no_dropout} securely implements the functionality in Figure \ref{fig:ideal}.
\end{theorem}
Full proof of this theorem can be found in Appendix~\ref{sec:no_dropout-proof}.

\begin{remark}\label{remark:communication}
    We note that the communication complexity between clients, the biggest bottleneck, can be improved beyond the naive implementation specified above. Recall that client-to-client communication is comprised exclusively of additive secret sharing of the client's secret. It can be advantageous to choose parameters in the encryption scheme so that the secret key is quite large. However, the communication complexity between parties will then grow accordingly as the size of each share is exactly the size of the secret key, in the naive implementation described above.

    To reduce communication costs, client $i$ holding a secret key share $\vs_i$ can, instead of additively sharing $\vs_i$ directly, sample a sequence of independent random seeds $r_1,\ldots,r_d$ that expand (using a PRG) to pseudorandom strings $\vy_1,\ldots,\vy_d$. By setting $\vy^* = \vs_i - \sum_{j=1}^d \vy_j$, the client can then send the seeds $\vr_1,\ldots,\vr_d$ to clients in the next cohort and $\vy^*$ to the server. The future client $j$ receiving a batch of $\vr_k$'s will then expand them to $\vy_k$'s and set their secret key $\vs_j = \sum \vy_k$. The server can use $\vy^*$ to ``correct'' the output of reveal by subtracting off $\vA\vy^*$ from the result, for the appropriately computed $\vA$.

    This results in client-to-client communication that is just $d\kappa$ bits, where $\kappa$ is the security parameter at the cost of sending an additional $|\vs|$ bits to the powerful and persistent server.
\end{remark}
\begin{figure}
\begin{tcolorbox}[enhanced, title=Server $\server$]
\begin{enumerate}
\item[] \hspace{-0.5cm} \textbf{Public parameters:}
\begin{itemize}
\item Clients' public keys (via PKI)
\item Vector length $\ell$, input domain $\mathbb{F}^\ell$.
\item Key sharing parameter $d$
\item Additive secret sharing scheme $\mathsf{AShare}$, threshold secret sharing scheme $\mathsf{TShare}$.
\item The $i$th instruction, $t_i=(\Mode_i, I_{i}, \{\vlambda^i_{k}\}_{k<i})$, from the program $P$ to be executed
\item A means of generating public random matrices $A_k$ indexed by corresponding round $k$
\item Discrete Gaussian distribution $D_\sigma$ for generating noise
\end{itemize}
\item[] \hspace{-0.5cm} For $i > 1$:
\item Send to $C_{i}$:
\begin{itemize}
    \item Encrypted shares of the key from $C_{i-1}$
\end{itemize}
\item Receive from each client of $C_{i}$:
\begin{itemize}
    \item Encrypted key resharings for $C_{i+1}$
    \item Receive $w_j^i$ from client $j$
\end{itemize}
\item Computation and Output:
\begin{itemize}
    \item Aggregate $w^i = \sum_{j\in [n]} w^i_j$
    \item If $\Mode_{i-1}=\Reveal$, compute and output the remainder $(w^{i-1}+\sum_{k<i-1} \vlambda^{i-1}_k w^k)$ modulo $T$
\end{itemize}
\end{enumerate}
\end{tcolorbox}
\caption{Server: No dropout resilience} \label{fig:server_no_dropout}
\end{figure}

\begin{figure}
\begin{tcolorbox}[enhanced, title=Client $C_{i,j}$]
\begin{enumerate}
\item[] \hspace{-0.5cm} \textbf{Public parameters:} 
\begin{itemize}
\item Clients' public keys (via PKI)
\item Vector length $\ell$, input domain $\mathbb{F}^\ell$.
\item Key sharing parameter $d$
\item Additive secret sharing scheme $\mathsf{AShare}$, threshold secret sharing scheme $\mathsf{TShare}$.
\item The $i$th instruction, $t_i=(\Mode_i, I_{i}, \{\vlambda^i_{k}\}_{k<i})$, from the program $P$ to be executed
\item A means of generating public random matrices $A_k$ indexed by corresponding round $k$
\item Discrete Gaussian distribution $D_\sigma$ for generating noise.
\end{itemize} 
\item[] \hspace{-0.5cm} \textbf{Private input:} Possibly an input vector $x_{i,j}:=I_{i}(D_{i,j})$ based on private data $D_{i,j}$.
\item Receive messages from $\server$:
\begin{itemize}
    \item Unless $i=1$, receive encrypted shares of key from $C_{i-1}$
    \end{itemize}
\item Local computation:
\begin{itemize}
    \item If $i=1$, sample uniformly random $s_{1,j}$ from $\mathbb{F}^{\ell}$.
    \item If $i>1$, decrypt shares and compute $s_{i,j} \leftarrow \sum_{r}s_{i-1,j}^r$
    \item Compute $(s^{1}_{i,j},\dots,s^{d}_{i,j})\leftarrow \mathsf{AShare}(s_{i,j})$
    \item If $\Mode_i=\Store$, generate and set $M_i = A_i$. Sample $g^i \gets D_\sigma$.%
    \item If $\Mode_i=\Reveal$, set $M_i = -\sum_{k<i}\lambda^i_{k} A_k$. Sample $g^i \gets\sum_{k<i}\lambda^i_kD_\sigma$.
    \item Compute $w^i_j = M_is_{i,j} + Tg^i$.%
\end{itemize}
\item Send to the server:
\begin{itemize}
    \item Send encrypted $(s^{1}_{i,j},\dots,s^{d}_{i,j})$ to $d$ random clients in $C_{i+1}$ via server
    \item Upload $w^i_j$ to server
\end{itemize}
\end{enumerate}
\end{tcolorbox}
\caption{Client $j$ in cohort $i$: No dropout resilience} \label{fig:client_no_dropout}
\end{figure}

\begin{figure}
\begin{tcolorbox}[enhanced, title=Server $\server$]
\begin{enumerate}
\item[] \hspace{-0.5cm} \textbf{Public parameters:}
\begin{itemize}
\item Clients' public keys (via PKI).
\item Input domain $\mathbb{F}^l$.
\item Key sharing parameter $d$.
\item Additive secret sharing scheme $\mathsf{AShare}$, threshold secret sharing scheme $\mathsf{TShare}$.
\item The $i$th instruction, $t_i=(\Mode_i, I_{i}, \{\lambda_{i,r}\}_{r<i})$, from the program $P$ to be executed
\item A list of dropped clients $\cD_i$ indexed by the $i$th cohort
\item A means of generating public random matrices $A_i$ indexed by corresponding round $i$
\item A pseudorandom generator $\mathsf{PRG}: \bbZ_q \rightarrow R_q^{\ell}$.
\item A sate $Z$ (initialized to 0) that aggregates missing secret shares.
\item Discrete Gaussian distribution $D_\sigma$ for generating noise.
\end{itemize}
\item[] \hspace{-0.5cm} For $i = 2$ to $N$:
\item Send to $C_{i}$:
\begin{itemize}
    \item Encrypted additive shares of key from $C_{i-1}$
    \item {\bf [Dropout Recovery]} Encrypted threshold shares from $C_{i-2}$ that recover the dropped clients in $C_{i-1}$
    \item {\bf [Remove Mask]} Encrypted threshold shares from non-dropped clients of $C_{i-1}$
\end{itemize}
\item Receive from each client in $C_{i}$:
\begin{itemize}
    \item Register dropped clients from $C_{i}$ in the global list $\cD_{i}$. 
    \item Encrypted key resharings for $C_{i+1}$
    \item Receive $w_j^i$ from client $j$ for every $j \in [n]\textbackslash \cD_i$
    \item Encrypted threshold shares of each additive key share to $C_{i+2}$
    \item Encrypted threshold shares of self-mask secret to $C_{i+1}$
    \item The decryption of all received threshold that server requested
\end{itemize}
\item Computation and Output:
\begin{itemize}
    \item Aggregate $\bar{w}^i = \sum_{j\in [n]\textbackslash\cD_i} w^i_j$
    \item Reconstruct missing key shares in $C_{i-1}$
\begin{itemize}
    \item Recover shares $s'_{k,j}$ that client $k$ sent to (dropped) client $j$ for all $k\in [n]\setminus \cD^{i-2}, j\in \cD^{i-1}$
    \item Update $Z \leftarrow Z + \sum_{k,j} s'_{k,j}$
\end{itemize}
\item Reconstruct self-mask of $C_{i-1}$
\begin{itemize}
    \item Recover $b_{i-1,j}$ for all non-dropped client $j$ in the $(i-1)$th cohort
    \item Compute $\mathsf{MASK}^{i-1}_j \leftarrow \mathsf{PRG}(b_{i-1,j})$
\end{itemize}
    \item If $\Mode_{i-1}=\Store$, set $w^{i-1} \leftarrow \bar{w}^{i-1} + \vA_{i-1}Z- \sum_{j \in [n] \textbackslash \cD_{i-1}} \mathsf{MASK}_{j}^{i-1}$
    \item If $\Mode_{i-1}=\Reveal$, compute and output the remainder $(\bar{w}^{i-1}+\sum_{k<i-1} \lambda_{i-1,k} (w^k-\vA_k Z)- \sum_{j \in [n] \textbackslash \cD_{i-1}} \mathsf{MASK}_{j}^{i-1})$ modulo $T$
\end{itemize}
\end{enumerate}
\end{tcolorbox}
\caption{Server: With dropout resilience} \label{fig:server_with_dropout}
\end{figure}

\begin{figure}[t]
\begin{tcolorbox}[enhanced, title=Client $C_{i,j}$]
\begin{enumerate}
\item[] \hspace{-0.5cm} \textbf{Public parameters:} 
\begin{itemize}
\item Clients' public keys (via PKI)
\item Input domain $\mathbb{F}^l$
\item Key sharing parameter $d$
\item Additive secret sharing scheme $\mathsf{AShare}$, threshold secret sharing scheme $\mathsf{TShare}$.
\item Chaperone parameter $h$
\item The $i$th instruction, $t_i=(\Mode_i, I_{i}, \{\lambda_{i,r}\}_{r<i})$, from the program $P$ to be executed
\item A list of dropped clients $\cD_i$ indexed by the $i$th cohort
\item A means of generating public random matrices $A_i$ indexed by corresponding round $i$
\item A pseudorandom generator $\mathsf{PRG}: \bbZ_q \rightarrow R_q^{m}$.
\item A sate $Z$ (initialized to 0) that aggregates missing secret shares.
\item Discrete Gaussian distribution $D_\sigma$ for generating noise.
\end{itemize}
\item[] \hspace{-0.5cm} \textbf{Private input:} Possibly an input vector $\vx_{i,j}:=I_{i}(D_{i,j})$ based on private data $D_{i,j}$.
\item Receive messages from $\server$:
\begin{itemize}
    \item Unless $i=1$, receive additive encrypted shares of key from $C_{i-1}$
    \item {\bf [Dropout Recovery]} Encrypted threshold shares from $C_{i-2}$ that recover the dropped clients in $C_{i-1}$
    \item {\bf [Remove Mask]} Encrypted threshold shares from non-dropped clients of $C_{i-1}$
    \end{itemize}
\item Local computation:
\begin{itemize}
    \item If $i=1$, sample uniformly random $s_{1,j}$ from $\mathbb{Z}_q^*$.
    \item If $i>1$, decrypt shares and compute $s_{i,j} \leftarrow \sum_{r}s_{i-1,j}^r$
    \item Compute $(s^{1}_{i,j},\dots,s^{d}_{i,j})\leftarrow \mathsf{AShare}(s_{i,j})$
    \item For each $s^{k}_{i,j}$, $k \in [d]$, compute $(t^{k,1}_{i,j},\dots,t^{k,h}_{i,j}) \leftarrow \mathsf{TShare}(s^{k}_{i,j})$
    \item Sample uniformly random secret $b_{i,j}$ from $\bbZ_q$ and compute $\mathsf{MASK}^{i}_j \leftarrow \mathsf{PRG}(b_{i,j})$
    \item Compute $(u^{1}_{i,j},\dots,u^{h}_{i,j}) \leftarrow \mathsf{TShare}(b_{i,j})$
    \item If $\Mode_i=\Store$, generate and set $M_i = A_i$. Sample $g^i \gets D_\sigma$.
    \item If $\Mode_i=\Reveal$, set $M_i = -\sum_{k<i}\lambda^i_{k} A_k$. Sample $g^i \gets\sum_{k<i}\lambda^i_kD_\sigma$.
    \item Compute $w^i_j = M_is_{i,j} + Tg^i + \mathsf{MASK}_j^i$.
\end{itemize}
\item Send to the server:
\begin{itemize}
    \item Send encrypted $(s^{1}_{i,j},\dots,s^{d}_{i,j})$ to $d$ random clients in $C_{i+1}$ via server
    \item For each $k\in [d]$, send encrypted $(t^{k,1}_{i,j},\dots,t^{k,h}_{i,j})$ to $h$ random chaperones of client $k$ in $C^{i+2}$ via server.
    \item Send encrypted $(u^{1}_{i,j},\dots,u^{h}_{i,j})$ to $h$ random chaperones of client $j$ in $C^{i+1}$ via server.
    \item The decryption of all received threshold that server requested.
    \item Upload $w^i_j$ to server
\end{itemize}
\end{enumerate}
\end{tcolorbox}
\caption{Client $j$ in cohort $i$: With dropout resilience} \label{fig:client_with_dropout}
\end{figure}

\subsection{Adding Dropout Resilience}
\label{sec:dropout}
We now describe how to augment the simple protocol present above to be resilient to client dropouts.

Recall that security and correctness of the protocol above effectively reduces to maintaining the invariant that at any point in time, the persistent secret key $\vs$ is safely additively secret shared across the current cohort of clients: clients are holding $\bar{\vs}_1,\ldots,\bar{\vs}_n$ that are uniformly random conditioned on $\vs = \sum_{i=1}^n \bar{\vs}_i$.

This is maintained by having each client re-share their share to some random subset of the next cohort: client $i$ holding $\bar{\vs}_i$ samples uniformly random $\bar{\vs}_i^1,\ldots,\bar{\vs}_i^d$ subject to $\sum_{j=1}^d \bar{\vs}_i^j=\bar{\vs}_i$ and sends each share $\bar{\vs}_i^j$ to a random client in the next cohort. Client $i$'s secret share $\bar{\vs}_i$ was in turn the result of summing the shares sent to it by clients in the previous cohort subject to $\bar{\vs}_i = \sum_{j=1}^{d'} \bar{\vs}'_j$.

We will augment this maintenance procedure with a simple mechanism to ensure that if any client drops out, their share can be recovered to maintain the persistent global secret key. Some delicacy is required to avoid compromising clients that send \emph{any} sensitive information (even if they are not able to complete a full protocol round).

Our approach is to associate with each client a random committee of \emph{chaperones} in some future cohort. The chaperones for client $i$ will effectively hold threshold secret shares of client $i$'s secret key so that if client $i$ drops out, a quorum of the chaperones can help the server to reconstruct $\bar{\vs}_i$. The server can then use $\bar{\vs}_i$ to ensure $\Reveal$ outputs are computed correctly (by subtracting off from the result $\vA\bar{\vs}_i$ for the appropriate matrix $\vA$). 

To enable the chaperones to do this, every client $j$ who sends a share $\vs'_j$ to client $i$ (so that client $i$'s share $\bar{\vs}_i = \sum_{j=1}^{d'} \vs'_j$) will additionally send \emph{threshold secret shares} of $\vs'_j$ to a randomly chosen set of chaperones in the cohort after $i$'s.

If client $i$ fails to send \emph{any} message specified by the protocol (recall that we assume all traffic is visible, but not its contents), then these chaperones will release their shares of $\vs'_j$ enabling the server to recover $\vs'_j$. Because this happens for all shares $\vs'_j$ sent to client $i$, the server can reconstruct $\bar{\vs}_i$. Provided that the chaperone committees do not have too many corrupted members, client $i$'s share $\bar{\vs}_i$ will remain perfectly secret.

The problem with what we have sketched so far, is that client $i$ may have sent an encryption of the input $\vx_i$ under $\bar{\vs}_i$ before dropping out. Therefore, if the procedure outlined above is followed, $\bar{\vs}_i$ can be used to decrypt client $i$'s private input $\vx_i$!

To avoid this, we introduce one last simple mechanism: client $i$ masks their encryption with a pseudorandom mask $\vy_i$. The short seed for this pseudorandom mask is then threshold secret shared with a random committee of chaperones in the next cohort. If the client successfully sends all messages, i.e.~does \emph{not} drop out, then the chaperones release their shares so the server can reconstruct the mask. If the client fails to send all messages as protocol defined, i.e.~the client is registered as having dropped out, the chaperones will not release the shares of the mask. Thus, provided there are not too many corrupt chaperones, the mask will remain pseudorandom and ensure the privacy of the client's input $\vx_i$.

\begin{remark}
    An alternative approach is to associate a single publicly known committee of chaperones with each client that will be used for all tasks above (instead of choosing a random committee for each underlying message in the dropout-free protocol). This committee could even be the same for all clients if one has good reason to believe that not too many will be corrupted. This may enable simpler implementation, albeit at the cost of making any one such committee easier to corrupt (as the membership is known in advance).

    Additionally, instead of automatically relying on the chaperones to release the mask for a client's encryption of their input, $\vx_i$. An alternate approach is to ask the client to send that directly in the next round, after checking all their messages were delivered successfully. Only if the client does not stay online would the chaperones need to reconstruct the mask. This way, clients who do not drop out will reduce the overall communication cost.
\end{remark}

\section{Benchmarking}\label{sec:benchmark}
In this section we discuss concrete parameter selection and performance costs for our protocol.  Given that RLWE encryption and decryption 
very fast in practice using FFT friendly parameters we focus on ciphertext expansion, and overall communication costs for clients.

\subsection{Parameter Selection}
We use the lattice estimator~\cite{DBLP:journals/jmc/AlbrechtPS15} to estimate the hardness of RLWE problem used and set parameters to have at least 128 bits of computational security. The secret distribution is a Gaussian with standard deviation $\sigma_s = \sqrt{2}\sigma$, and the noise distribution used in fresh encryptions is Gaussian with standard deviation $\sigma_n = 2\sigma\sqrt{r+1}$, where $\sigma = 3.2$ is the stdev used to estimate RLWE security, and $r$ is the number of rounds/releases in the protocol. Since we only apply additive homomorphic operations on ciphertexts, we can track the $l_\infty$ norm on the coefficient embedding of the error polynomial, as well as the plaintext. 

 Security of the our scheme relies on a variant of
Kim et al.~\cite{C:KLSS22} %
assumption  which proves that by using a uniformly random secret and Gaussian error with standard deviation $\sigma_n$ as set above, the resulting aggregation protocol is as secure as standalone HE scheme with $\sigma = 3.2$.
Note that the specific choice of $\sigma = 3.2$ as the error distribution for standalone homomorphic encryption schemes is suggested by the Homomorphic Encryption Standard~\cite{HomomorphicEncryptionSecurityStandard}, and is widely accepted and used in practice.

\subsection{Communication costs}
Table~\ref{tab:params} shows some parameters for common settings of the protocol. We find parameters by minimizing the dominant communication costs (bits sent to server) while doing a grid search over secure parameters. It can be observed that as soon as vector length is $> 1000$
ciphertext packing pays off, and ciphextext expansion stays within small single digits (2-5x). By ciphertext packing we mean encoding several 
entries of the vectors to be encrypted in the same coefficient of a ciphertext.
This gives flexibility when finding parameters as it allows to use the plaintext domain optimally. This explain the large values of $N$ and $q$ in the table.
An important optimization well-known in practice is to drop unused coefficients of a ciphertext.

We now turn our attention to the prefix sum application. In Figure~\ref{fig:client_comm_l}, we show how per-client communication scales with input vector length, compared to a baseline insecure protocol where plaintexts are submitted in the clear. As shown in the plot, the communication overhead is small even for small $\epsilon$ (privacy parameter). The plaintext domain has to account for
differentially private noise. As the corresponding distribution is a centered
Gaussian, this is not a huge increase over the requirement that the plaintext domain has to fit the sum of $n$ clients' contributions, i.e. $n$, as every input is a binary vector in this case.

We do not report the client-to-client communication cost as this is (a) comparatively much smaller and (b) invariant regardless of how the encryption scheme is instantiated (see Remark~\ref{remark:communication}).

\begin{figure}\label{fig:comm-plot}
    \centering
    \includegraphics[width=0.5\linewidth]{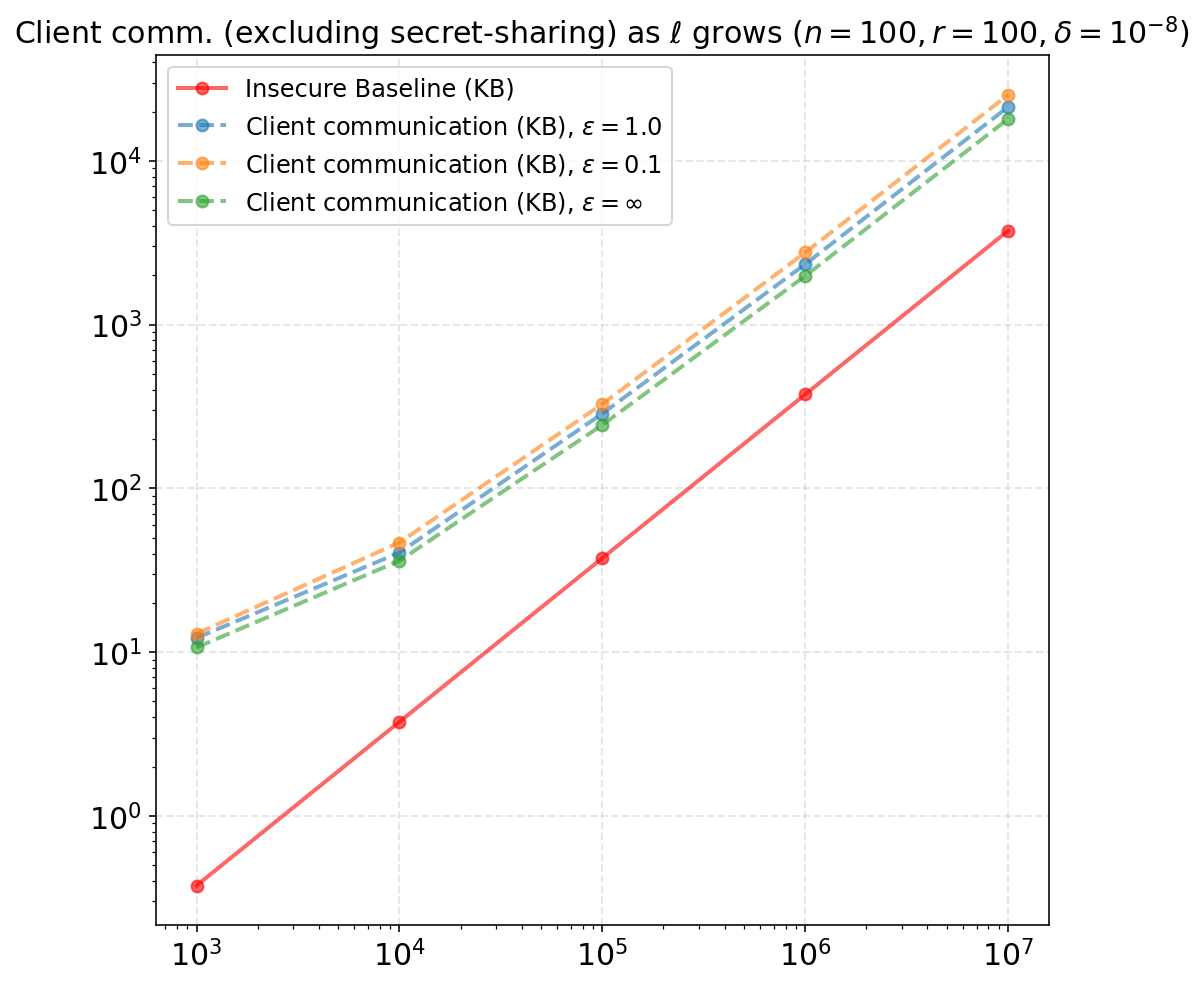}
    \caption{Client communication to server for several vector lengths for our protocol ($1000$ rounds), compared to the baseline where clients submit a vector in $\{0,1\}^\ell$ in the clear. }
    \label{fig:client_comm_l}
\end{figure}

\begin{table}\label{fig:comm-table}
    \centering
\begin{tabular}{ccccccc}
\toprule
 $n$    & $\ell$   &   $N$ &   $q$ &   packing factor & client communication   & ciphertext expansion   \\
 \midrule
 $10^3$ & $10^3$   &  2048 &    44 &                1 & 16.76 KB               & 8.38x                  \\
 $10^5$ & $10^3$   &  2048 &    54 &                1 & 20.57 KB               & 10.29x                 \\
 $10^7$ & $10^3$   &  4096 &    64 &                1 & 40.77 KB               & 20.38x                 \\
 $10^3$ & $10^5$   &  4096 &    96 &                3 & 449.16 KB              & 2.25x                  \\
 $10^5$ & $10^5$   &  4096 &    87 &                2 & 588.29 KB              & 2.94x                  \\
 $10^7$ & $10^5$   &  4096 &   103 &                2 & 696.49 KB              & 3.48x                  \\
 $10^3$ & $10^7$   & 16384 &   434 &               16 & 34.88 MB               & 1.74x                  \\
 $10^5$ & $10^7$   & 16384 &   413 &               12 & 43.87 MB               & 2.19x                  \\
 $10^7$ & $10^7$   & 16384 &   417 &               10 & 52.98 MB               & 2.65x                  \\
\bottomrule
\end{tabular}
    \label{tab:params}
    \caption{Parameters for some common settings (number of rounds is 1000 in all cases, and input domain is $[2^{16}]^\ell$).}
    
\end{table}

\bibliographystyle{ACM-Reference-Format}
\bibliography{main}

\newpage
\appendix

\section{Preliminaries}

\subsection{Cryptographic Building Blocks}

\subsubsection{Lattices, Rings, and RLWE Encryption}
\begin{definition}[Decisional LWE assumption]
    Given a prime $q$, a matrix $A$ sampled uniformly from $\mathbb{Z}_q^{M\times N}$, a vector $s$ uniformly sampled from $\mathbb{Z}_q^{N}$, an error vector $e \in \mathbb{Z}_q^{N}$ sampled from $\phi$. We say that Decisional LWE is hard if $(A, As+e)$ is computationally indistinguishable from the uniform.
\end{definition}

\begin{definition}[LWE encryption]
    Let $A$ be sampled uniformly from $\mathbb{Z}_q^{M\times N}$ and $\phi$ be an error distribution in which LWE is hard. LWE encryption scheme consists of the following three algorithms:
    \begin{itemize}
        \item $\mathsf{LWEGen}(1^\lambda) \rightarrow (s)$: sample a random vector $s$ from $\mathbb{Z}^{N}_q$, sample $e \leftarrow \phi$.
        \item $\mathsf{LWEEnc}(A, s, x) \rightarrow c$: compute $c = As + Te + x \mod q$, where $x \in \mathbb{Z}_T^M$, $T$ is coprime to $q$.
        \item $\mathsf{LWEDec}(A, s, c) \rightarrow x$: compute $x = (c - As) \mod T$.
    \end{itemize}
    We say that LWE encryption is secure is secure if it has CPA security. Note that if Decisional LWE assumption holds, then LWE encryption is secure.
\end{definition}

\subsubsection{Secret Sharing}

\begin{definition}[Additive secret sharing]
An additive secret-sharing scheme consists of the following algorithm:
\begin{itemize}
    \item $\mathsf{AShare}(sk) \rightarrow (x_{1},\dots,x_{n})$: take in a secret $sk$, generate randomness $\rho$, output $n$ random messages such that $\sum_{i=1}^{n} x_i = sk$. 
\end{itemize}
\end{definition}

\begin{definition}[Threshold secret sharing]
A threshold secret-sharing scheme consists of the following algorithm:
\begin{itemize}
    \item $\mathsf{TShare}(sk) \rightarrow (x_{1},\dots,x_{m})$: take in a secret $sk$, generate randomness $\rho$, output $m$ random messages to help with the reconstruction of $sk$.
    \item $\mathsf{TRec}(\{x_{j}\}_{j\in S, |S| \geq t}) \rightarrow sk$: take in at least $t$ threshold secret shares, output $sk$.
\end{itemize}
\end{definition}
We say that the threshold secret sharing scheme is secure if the probability that, given less than $t$ threshold shares, an adversary can recover $sk$ is negligible.

\section{Fully-Synchronous (no dropouts) semi-honest security}\label{sec:no_dropout-proof}

We use a variant case of the Hint-MLWE problem of Kim et al.~\cite{C:KLSS22}. Kim et al.~additionally allow leakage on the MLWE secret, $s$, but we do allow the adversary this. On the other hand, we do assume the secret key is uniformly random (in contrast to Gaussian). Additionally, our variant considers multiplying the noise by a factor of $T$ (where $T$ is coprime to the modulus $q$).

\begin{definition}[MLWE~\cite{C:KLSS22}]
\label{def:MLWE}
Let $d,m,q$ be positive integers, and $\chi$ be a distribution over $R^{d+m}$. Then, the goal of the Module-LWE (MLWE) problem is to distinguish $(\vA, \vu)$ from $(\vA, \vA\vs+\ve)$ for $\vA \ugets R_{q}^{m \times d}, \vu \leftarrow \mathcal{U}(R_{q}^{m})$, $\ve \leftarrow \chi$, and $\vs \ugets R_q^d$. We say that a PPT adversary $\mathcal{A}$ has advantage $\varepsilon$ in solving $\MLWE_{R,d,m,q,\chi}$ if \[\Pr[\cA(\vA,\vA\vs+\ve)=1] - \Pr[\cA(\vA,\vu)=1] \geq \varepsilon.\]
    
\end{definition}

\begin{definition}[Hint-MLWE~\cite{C:KLSS22}]
\label{def:HintMLWE}
    Let $d,m,\ell$ be positive integers, $\chi,\xi$ be distributions over $R^{m}$, $\chi$ a distribution over $R^d$. The Hint-MLWE problem, denoted by $\HintMLWE^{\ell,\xi}_{R,d,m,q,T,\chi,\chi'}$ asks an adversary $\cA$ to distinguish the following two cases:
    \begin{enumerate}
        \item $(\vA,\vA\vs+T\ve, \ve+\vf_1,\ldots,\ve+\vf_\ell)$,
        \item $(\vA,\vu,\ve+\vf_1,\ldots,\ve+\vf_\ell)$;
    \end{enumerate}
    where in both distributions above $\vA\ugets R^{m\times d}_q$, $\vs\ugets \chi'$, $\ve\gets \chi$, $\vf_i\gets \xi$ for all $i\in[\ell]$, and $\vu\ugets R^m_q$.

    We take $\HintMLWE^{\ell,\sigma_2}_{R,d,m,q,T,\sigma_1}$ to be the case where $\xi$ is a spherical Gaussian distribution with width $\sigma_2$ and $\chi$ is a spherical Gaussian distributions with width $\sigma_1$, and $\chi'$ is uniformly random.
\end{definition}
We verify that Kim et al.'s proof of the hardness of Hint-MLWE can similarly be adapted to our variant following the observations of Bell et al.~for HintLWE~\cite{USENIX:BGLLMY23} and noticing that if the secret isn't leaked then conditional sampling need not be invoked on the secret key.

\begin{definition}[Smoothing parameter~\cite{1366257}]
For an $n$-dimensional lattice $\Lambda$ and positive real $\varepsilon>0$, the smoothing parameter $\eta_{\varepsilon}(\Lambda)$ is the smallest $s$ such that $\rho_{1/s}(\Lambda^{*}\textbackslash\{\mathbf{0}\})\leq\varepsilon$.
    
\end{definition}

\begin{theorem}[\cite{C:KLSS22}]
    Let $d,k,m,q,\ell$ be positive integers.
    For any $T$ coprime to $q$, $\ell>0$, and $\sigma_1,\sigma_2,\sigma>0$ such that $\frac{1}{\sigma^2}\ge 2\left(\frac{1}{\sigma_1^2}+\frac{\ell}{\sigma_2^2}\right)$. If $\sigma\ge \sqrt{2}\eta_\varepsilon(\bbZ^n)$ for $0<\varepsilon\le 1/2$, then there exists an efficient reduction from $\MLWE_{R,d,m,q,\sigma}$ to $\HintMLWE^{\ell,\sigma_2}_{R,d,m,q,T,\sigma_1}$.
\end{theorem}

From this we derive a sanity check for our security proof that says that the leakage of the reveal operations does not reveal anything about the encrypted state held by the server. The full required for the proof amounts to incorporating a hybrid argument with this observation. Note that here we give a general proof for the high dimensional case. Our protocol actually works under Hint-RLWE assumption, where $d$ is set to 1.

\zxnote{Here is the general proof. Our protocol actually is in the case $d=1$. I will leave this session unchanged but add some comments.}

\begin{corollary}\label{cor:linear-leakage}
    Assuming the hardness of $\HintMLWE^{\ell,\sigma_2}_{R,d,m,q,T,\sigma_1}$ \mnote{fix params}, for any $\vlambda^{(1)},\ldots,\vlambda^{(\ell)}\in \bbZ_q^m$, the following distributions are indistinguishable for any PPT adversary
    \begin{enumerate}
        \item $(\vA_1,\ldots,\vA_r, \vA_1\vs+T\ve_1,\ldots,\vA_r\vs+T\ve_r, \sum_{i=1}^r\vlambda^{(1)}_i(T\vf^{(1)}_i-\vA_i\vs),\sum_{i=1}^r\vlambda^{(\ell)}_i(T\vf^{(\ell)}_i-\vA_i\vs))$,
        \item $(\vA_1,\ldots,\vA_r, \vu_1,\ldots,\vu_r, \sum_{i=1}^r\vlambda^{(1)}_i(T\vf^{(1)}_i+T\ve_i-\vu_i),\sum_{i=1}^r\vlambda^{(\ell)}_i(T\vf^{(\ell)}_i-T\ve_i-\vu_i))$;
    \end{enumerate}
    where in both distributions above $\vA_i\ugets R^{m\times d}_q$, $\vs\gets D_{d,\sigma_1}$, $\ve\gets D_{m,\sigma_1}$, $\vf_i^{(j)}\gets D_{\sigma_2}$, and $\vu_i\ugets R^m_q$ for all $i\in[r], j\in [\ell]$.
\end{corollary}

\begin{proof}
    
This follows by observing that each distribution is a linear function from the corresponding HintMLWE distribution. Namely, if we take
\[ \vA := \left[\begin{array}{c}
     \vA_1\\\cdots\\\vA_n\end{array}\right]\quad\&\quad 
     \vu := \left[\begin{array}{c}
     \vu_1\\\cdots\\\vu_n\end{array}\right]
     \quad\&\quad\ve := \left[\begin{array}{c}
     \ve_1\\\cdots\\\ve_n\end{array}\right]\quad\&\quad \vf^{(i)} :=\left[\begin{array}{c}
     \vf^{(i)}_1\\\cdots\\\vf^{(i)}_n\end{array}\right] \]
for all $i \in [\ell]$ and define the operator
\[ M(\vlambda):= \left[ \vlambda_1 I |\cdots \vlambda_n I\right] = \vlambda^\top \otimes I\]
then we can equivalently formulate the above distributions as
\begin{enumerate}
    \item $\big(\vA,\vA\vs+T\ve,M(\lambda^{(1)})(T\vf^{(1)}-\vA\vs),\ldots,M(\lambda^{(\ell)})(T\vf^{(\ell)}-\vA\vs)\big)$,
    \item  $\big(\vA,\vu,M(\lambda^{(1)})(T\vf^{(1)}+T\ve-\vu),\ldots,M(\lambda^{(\ell)})(T\vf^{(\ell)}+T\ve-\vu)\big)$.
\end{enumerate}

The corollary then follows from the simple reduction $F$ that given $(\vA,\vb,\vy^{(1)},\ldots,\vy^{(\ell)})$ and $\vlambda^{(1)},\ldots,\
\vlambda^{(\ell)}$ simply outputs $$(\vA,\vb,M(\vlambda^{(1)})(T\vy^{(1)}-\vb),\ldots,M(\vlambda^{(\ell)})(T\vy^{(\ell)}-\vb)).$$
To see this suffices note that for $\vy^{(i)}=\ve+\vf^{(i)}$, \[M(\vlambda^{(i)})(\vy^{(i)}-\vb) = \left\{\begin{array}{ll}
    M(\vlambda^{(i)})(T\vf^{(i)}-\vA\vs) &\mbox{if } \vb = \vA\vs+T\ve,  \\
     M(\vlambda^{(i)})(T\vf^{(i)}+T\ve-\vu) & \mbox{if }\vb = \vu.
\end{array}\right.\]
\end{proof}

However, as mentioned, this corollary alone is not quite enough for our proof. We need to show that not only does the encrypted state remain secure under the $\Reveal$ operations but moreover that the $\Store$ operations themselves remain secure given these releases: encrypting under correlated keys remain secure under leakage.

\begin{lemma}\label{lem:key_no_dropout}
    Assuming the hardness of $\HintMLWE^{\ell,\sigma_2}_{R,d,m,q,T,\sigma_1}$, for any $\vlambda^{1},\ldots,\vlambda^{\ell} \in \bbZ_q^m, \vx_1,\ldots,\vx_t \in R_q^m$ the following distributions are indistinguishable for any PPT adversary
    \begin{enumerate}
    \item $\big(\vA_1,\ldots, \vA_r,\vA_1\vs+T\ve_1+\vx_1,\ldots,\vA_r\vs+T\ve_r+\vx_r,\sum_{i=1}^r\vlambda_i^{1}(T\vf_i^{1}-\vA_i\vs-\vx_i)+y_1,\ldots,\sum_{i=1}^r\vlambda_i^{\ell}(T\vf_i^{\ell}-\vA_i\vs-\vx_i)+y_{\ell}\big)$,
    \item  $\big(\vA_1,\ldots, \vA_r,\vu_1,\ldots,\vu_r,\sum_{i=1}^r\vlambda_i^{1}(T\vf_i^{1}+T\ve_i-\vu_i)+y_1,\ldots,\sum_{i=1}^r\vlambda_i^{\ell}(T\vf_i^{\ell}+T\ve_i-\vu_i)+y_{\ell}\big)$;
\end{enumerate}
  where in both distributions above $\vA_i\ugets R^{m\times d}_q$, $\vs\gets D_{d,\sigma_1}$, $\ve_i\gets D_{m,\sigma_1}$, $\vf^{(j)}_i\gets D_{\sigma_2}$ for all $i\in[r], j\in [\ell]$, and $\vu_1,\ldots,\vu_r\ugets R^m_q$ and we define $y_k := \sum_{i = 1}^r \vlambda_i^{k}\vx_i$ for $k\in[\ell]$.
\end{lemma}

The proof of this lemma combines the simple transformation above with a hybrid argument.

\begin{proof}
    We begin by defining the hybrids $H_0,\ldots,H_r$. In particular, $H_i$ is the distribution
    \[
    \left(\vA_1,\ldots, \vA_r,\vu_1,\ldots,\vu_i,
    \vA_{i+1}\vs + \vx_{i+1} + T\ve_{i+1},\ldots,\vA_{r}\vs + \vx_{r} + T\ve_{r}, z_1,\ldots z_\ell\right)
    \]
    where for $k\in[\ell]$,
    \[z_k =  \sum_{i=1}^{r}\vlambda_i^{(k)}\left(-\sum_{j=1}^i \vu_j-\sum_{j=i+1}^r \vA_j\vs + T\vf_i^{k} + T\sum_{j=1}^i \ve_j-\sum_{j=i+1}^r \vx_j\right) + y_k\]
    and $\vu_1,\ldots,\vu_r\ugets R^m_q, \vA_i\ugets R^{m\times d}_q, \vs\gets D_{d,\sigma_1}, \ve_i\gets D_{m,\sigma_1}, \vf^{(j)}_i\gets D_{\sigma_2}$ for $i\in [r], j\in [\ell]$.

    Notice that $H_0$ coincides with distribution 1 and $H_r$ coincides with distribution 2. Now, we argue that for any $i\in [r]$, $H_{i-1}\approx H_i$ by the following reduction from $\HintMLWE$ (Definition~\ref{def:HintMLWE}).

    We can build a reduction $\cB$ that takes input $ (\vA,\vb,\vw^{(1)},\ldots,\vw^{(\ell)})$ from $\HintMLWE$ such that it outputs
    \[
     \left(\vA_1,\ldots, \vA_r,\vu_1,\ldots,\vu_{i-1},\vv_i,\vA_{i+1}\vs+\vx_{i+1}+T\ve_{i+1},\ldots,\vA_r\vs+\vx_r+T\ve_r,\hat{z}_1,\ldots, \hat{z}_\ell\right),
    \]
    where $\vv_i = \vb+\vx_i$ and
    \[
    \hat{z}_k = \sum_{i=1}^{r}\vlambda_i^{k}\left(-\vb + \vw^{(k)} -\sum_{j=1}^{i-1} \vu_j-\sum_{j=i+1}^r \vA_j\vs + T\sum_{j=1}^{i-1} \ve_j-\sum_{j=i}^r \vx_j\right)+y_k
    \]

    Consider that $\vb = \vA\vs + T\ve$, $\vw^{(i)} = T\ve + T\vf^{(i)}$ for $i \in [\ell]$, $\ve\gets D_{m,\sigma_1}, \vf^{(i)}\gets D_{\sigma_2}$. We rewrite using fresh variable names $\vA_i = \vA, \ve_i = \ve, \vf_i^{k} = \vf^{(i)}$. As a result, we have $\vv_i = \vA_i\vs + \vx_i + T\ve_{i}$ and

    \[
    \hat{z}_k = \sum_{i=1}^{r}\vlambda_i^{k}\left( -\sum_{j=1}^{i-1} \vu_j-\sum_{j=i}^r \vA_j\vs + T\vf_i^{k}+ T\sum_{j=1}^{i-1} \ve_j-\sum_{j=i}^r \vx_j\right)+y_k
    \]

    So the output of the reduction is the same as $H_{i-1}$.

    On the other hand, consider that $\vb = \vu, \vw^{(i)}= T\ve + T\vf^{(i)}$ for $i \in [\ell]$, $\ve\gets D_{m,\sigma_1}, \vf^{(i)}\gets D_{\sigma_2}$. We rewrite using fresh variable names $\vu_i = \vu, \ve_i = \ve, \vf_i^{k} = \vf^{(i)}$. As a result, we have $\vv_i = \vu_i + \vx_i$ and 

    \[
    \hat{z}_k = \sum_{i=1}^{r}\vlambda_i^{k}\left(-\sum_{j=1}^{i} \vu_j-\sum_{j=i+1}^r \vA_j\vs + T\vf_i^{k} + T\sum_{j=1}^{i} \ve_j-\sum_{j=i}^r \vx_j\right)+y_k
    \]

    Since $\vb = \vu$ is uniformly random, $\hat{z}_k$ follows exactly the distribution $H_i$. By the hardness assumption of $\HintMLWE$, for any $i \in [r], H_{i-1}$ and $H_i$ are indistinguishable for any PPT adversary. In particular, $H_{0}$ and $H_r$ are indistinguishable for any PPT adversary, which proves the lemma.
    
\end{proof}
\begin{theorem}
    Assuming that a semi-honest adversary corrupts at most an $\gamma$-fraction of any user cohort, in addition to the server, 
    The protocol given in Figures~\ref{fig:server_no_dropout},\ref{fig:client_no_dropout} securely implements the functionality in figure \ref{fig:ideal}.
\end{theorem}
The proof follows more or less immediately from Lemma~\ref{lem:key_no_dropout}.

\begin{proof}[Sketch]
    For simplicity, we assume secure point-to-point channels between clients in successive cohorts. This assumption is easily relaxed in usual manner by introducing key infrastructure. Also for simplicity we assume the inputs are statically chosen, although the proof can easily be modified to securely simulate a reactive version of the functionality.

    {\bf The lazy protocol.} In this proof we consider, without loss of generality, a \emph{lazy protocol} where all linear function evaluation is pushed to the time of reveal. The lazy server simply appends the sum of the clients' $\Store$ messages to it's state, instead of also adding a linear combination of the prior states. Then, when it is time to reveal the server and clients recursively compute the appropriate linear function, $\bar{\vlambda}^i$, consistent with all real linear operations that were to have taken place up to this point, $\vlambda^1,\ldots,\vlambda^{i-1}$. By linearity, this too is just another linear function that can be computed by copying the state and performing the sequence of necessary linear evaluations or, alternatively, simply evaluating the with $\bar{\vlambda}^i$ where $\bar{\vlambda}^i_j:= \sum_{j=k_1<k_2<\cdots<k_\ell=i}\vlambda^{k_\ell}_{k_{\ell-1}}\vlambda^{k_{\ell-1}}_{k_{\ell-2}}\cdots \vlambda^{k_3}_{k_2}\vlambda^{k_2}_{k_1}$ for $j<i$ and $\bar{\vlambda}^i_i = 1$.

    For example if the sum of inputs at round $i$ is $x_i$, then the real server has state
    \[
       \vv= (\underbrace{x_1}_{v_1},\underbrace{x_2 + \lambda^2_1v_1}_{v_2}, \underbrace{x_3+\lambda^3_2v_2+\lambda^3_1v_1}_{v_3},\ldots)
    \]
    and when asked to reveal simply outputs $v_i$.

    In the lazy protocol, the state is simply
    \[ \bar{\vv} = (x_1,x_2,x_3,\ldots) \]
    and when asked to reveal the output is computed as $\langle\bar{\vlambda}^i,\bar{\vv}\rangle = \sum_{j=1}^i \bar{\vlambda}^i_j x_j$.

    {\bf Simulating client-to-client communication.}
    We begin by noting that if an adversary controls at most an $\gamma$ fraction of the next cohort, $C^{i+1}$, the probability that any specific client in cohort $C^i$ who sends messages to $d$ i.i.d.~randomly chosen clients in $C^{i+1}$ fails to send a message to an honest client is $\gamma^d$. If there are a total of at most $r$ cohorts, each containing $n$ clients, then for $d \geq \frac{\log(1/\delta) + \log(2nr)}{1-\gamma}$ every honest client sends a message to some other honest client and receives at least one message from some other honest client with probability at least $1-2nre^{d(\gamma-1)}\ge 1-\delta$. 
    
    In this case, because messages from honest clients to other clients are perfect additive secret shares as we assume point-to-point channels between clients, the adversary's view of such messages is uniformly random and independent of all other communication in the protocol seen by the adversary. Thus, we can focus exclusively on simulating messages to the corrupt server (which comprises all other messages from honest clients seen by the server). Next, we describe how to simulate these messages.

    {\bf Simulating messages to the adversarial server.}
    Let $\hat{\vv_i}$ denote the state of the idealized functionality at round $i$ if only adversarial inputs are incorporated. Let $\hat{\vv_i}'$ denote the state of the idealized functionality using only honest outputs at round $i$, then the state at round $i$ (incorporating both honest and adversarial inputs) is simply $\hat{\vv_i}+\hat{\vv_i}'=\vv_i$. Recall, additionally, that if there is a $\Reveal$ operation at round $i$, then the idealized functionality will output $\vv_i$ at the next round.

    The simulator is as follows:
    \begin{enumerate}
        \item The messages sent by honest parties when $\Mode_i=\Store$ are simply uniformly random.
        \item When $\Mode_{i}=\Reveal$, the simulator additionally sends random values (in the next round) such that when summed with the contribution from server, $y_i = \langle\bar{\vlambda}^i,\bar{\vv}\rangle$, where $\bar{\vv} = (x_1,x_2,x_3,\ldots)$ is the server's state at time $i$, and the messages from the corrupt parties, the result is $y_i+T\sum_{j=1}^{t'_i} \langle\bar{\vlambda}^i,\vf^{i,j}\rangle$ where $\vf^{i,j}$ are sampled according to $ D_{\sigma_1} + D_{\sigma_2}$.%
    \end{enumerate}
    We now argue that this simulation is indistinguishable from the real protocol using a hybrid argument. However before continuing, we set up some notation that will be consistent across the hybrid argument. Suppose there are $r$ instructions with $\Mode = \Store$ and $\ell$ instructions with $\Mode = \Reveal$. We denote the view of the adversarial server (rearranged) as $$(\vA_1, \dots,\vA_{r},\alpha^1_1,\ldots,\alpha^1_{t_1},\ldots, \alpha^r_1,\ldots,\alpha^r_{t_r},\beta^1_1,\ldots,\beta^1_{t'_1},\ldots,\beta^\ell_1\ldots,\beta^\ell_{t'_\ell})$$ where $\alpha^i_j$ denotes the $\Store$ message sent by the $j$th honest client in round $i$ (here the honest clients are indexed from $1$ to $t_i$) and $\beta^i_j$ denotes the message sent by the $j$th honest client in round $i$ if $\Mode_{i-1}=\Reveal$ (here the honest clients are indexed from $1$ to $t_i'$). Caution: the distribution of these variables will depend on the specific hybrid.

    So the simulator sets all $\alpha^i_j$ to be uniformly random and independent. Similarly, $\beta^i_j$ are uniformly random conditioned on $\sum_{j=1}^{t'_i} \beta^i_j = y_i+T\sum_{j=1}^{t'_i} \langle\bar{\vlambda}^i,\vf^{i,j}\rangle - \vc_i$, where $\vc_i$ is the message from the dishonest parties.

    We now define the sequence of hybrids:
    
    {\bf Hybrid 0.} $H_0$ is the adversary's real view. We can rewrite how this is sampled as follows:
    \begin{enumerate}
        \item Sample $\vA_i \ugets R_q^{m\times d}$ for all $i \in [r]$.
        \item Sample $\vs\ugets R_q^m$.
        \item Sample $\ve_j^{i} \gets D_{m,\sigma_1}$ for all $i\in [r], j\in [t_i]$ and $\vf_{k,j}^{i}\gets D_{\sigma_2}$ for all $i\in [\ell], j \in [t_i'], k \in [r]$.%
         \item For all $i \in [r], \vs_j^{i}$ are uniformly random variables conditioned on $\sum_{j=1}^{t_i} \vs_j^{i}+ \sum_{k \text{ corrupt}} \vs_k^{i} = \vs$, where $\vs_k^{i}$ for corrupt $k$ are computed according to the protocol.
        \item For all $j \in [t_i]$, set $\alpha^i_j = \vA_i\vs_j^i+T\ve_j^{i} + x_j^{i}$, where $x_j^{i}$ is the secret input of the $j$th honest party in round $i$.
        \item For all $j \in [t_i']$, set $\beta^i_j = \sum_{k=1}^{r}\bar{\vlambda}^{i}_k(T\vf^{i}_{k,j}-\vA_k\vs_j^{i})$, where $\bar{\vlambda}^{i} \in \bbZ_q^m$ is the weights applied in round $i$.
        \item $\hat{\alpha}^i_k,\hat{\beta}^i_k$ for corrupt $k$ are computed according to protocol and prior messages.
    \end{enumerate}
    
    {\bf Hybrid 1.} In this hybrid, we can equivalently sample this by first sampling uniformly random $\alpha_j^i \ugets R_q^m$ for all $i \in [r]$ and $j\geq 2$. Then we set $\alpha_1^i = \sum_{j=1}^{t_i} \vA_i s^{i}_j + T\ve^{i}_j + x^{i}_j - \sum_{j=2}^{t_i} \alpha_j^i$ for all $i \in [r]$.

    Notice that all clients need to follow the protocol regardless of whether honest or not. So the invariant $\sum_{j=1}^{t_i} \vA_i s^{i}_j + \sum_{k \text{ corrupt}} \vA_is_k^i = \vA_i \vs$ always holds. So we can rewrite $\alpha_1^i$ as
    \[
    \vA_i\vs + T\ve_i + \vx_i - \vc_i - \sum_{j=2}^{t_i} \alpha_j^i
    \]
    where $\ve_i =  \sum_{j=1}^{t_i} \ve^{i}_j$, $\vx_i = \sum_{j=1}^{t_i} \vx^{i}_j$, and $\vc_i = \sum_{k \text{ corrupt}} \vA_is_k^i$.
    
    To show this hybrid is indistinguishable from {\bf Hybrid 0}, we define the following partial hybrids $H_i^1, \dots, H_i^{t_i}$ for each $i \in [r]$. In particular, $H_i^j$ is the distribution:
    \[
    \left(\vA_i,\sum_{k=1}^j \vA_i\vs^i_k + \vx^i_k +T\ve_k^i - \sum_{k=2}^{j}\vu_k, \vu_2,\ldots,\vu_j,
    \vA_{i}\vs^i_{j+1} + \vx^i_{j+1} + T\ve^i_{j+1},\ldots,\vA_{i}\vs^i_{t_i} + \vx^i_{t_i} + T\ve^i_{t_i}\right)
    \]

    where $\vu_2,\ldots,\vu_{t_i}\ugets R^m_q, \vA_i\ugets R^{m\times d}_q, \vs\ugets R_q^d, \ve^i_{1},\dots,\ve^i_{t_i}\gets D_{m,\sigma_1}$, and $\vs_j^{i}$ are uniformly random variables conditioned on $\sum_{j=1}^{t_i} \vs_j^{i}+ \sum_{k \text{ corrupt}} \vs_k^{i} = \vs$, where $\vs_k^{i}$ for corrupt $k$ are computed according to the protocol.

    Notice that $H_i^1$ coincides with the partial view in {\bf Hybrid 0} and $H_i^{t_i}$ coincides with the partial view in {\bf Hybrid 1}. Now, we argue that for any $j\in \{2,\dots,t_i\}$, $H^{j-1}_{i}\approx H^{j}_i$ by the following reduction from $\MLWE$ (Definition~\ref{def:MLWE}).

     We can build a reduction $\cB$ that takes input $ (\vA,\vb)$ from $\MLWE$ such that it outputs
    \[
     \left(\vA_i,\vz, \vu_2,\ldots,\vu_{j-1},\vv_j,\vA_{i}\vs^i_{j+1} + \vx^i_{j+1} + T\ve^i_{j+1},\ldots,\vA_{i}\vs^i_{t_i} + \vx^i_{t_i} + T\ve^i_{t_i}\right),
    \]
    where $\vv_j = \vb+\vx^i_j$ and
    \[
    \vz = \sum_{k=1}^{j} \vA_i\vs^i_k + \vx^i_k +T\ve_k^i - \sum_{k=2}^{j-1}\vu_k - \vb - \vx_j^i
    \]

    Consider that $\vb = \vA\vs + T\ve$, $\ve\gets D_{m,\sigma_1}$. We rewrite using fresh variable names $\vA_i = \vA, \vs_{j}^i = \vs, \ve^i_j = \ve$. As a result, we have $\vv_j = \vA_i\vs_j^i + \vx_j^i + T\ve^{i}_j$ and

    \[
    \vz = \sum_{k=1}^{j-1} \vA_i\vs^i_k + \vx^i_k +T\ve_k^i - \sum_{k=2}^{j-1}\vu_k
    \]

    So the output of the reduction is the same as $H_i^{j-1}$.

    On the other hand, consider that $\vb = \vu$. We rewrite using fresh variable names $\vu_j = \vu$. As a result, we have $\vv_j = \vu_j + \vx_j^i$ and 
    
    \[
    \vz = \sum_{k=1}^{j} \vA_i\vs^i_k + \vx^i_k +T\ve_k^i - \sum_{k=2}^{j}\vu_k - \vx_j^i
    \]

    Since $\vb = \vu$ is uniformly random, $\vz$ follows exactly the distribution $H_i^j$. By the hardness assumption of $\MLWE$, for any $j\in \{2,\dots,t_i\}$, $H_i^{j-1}$ and $H_i^j$ are indistinguishable for any PPT adversary. In particular, $H_{i}^1$ and $H_i^{t_i}$ are indistinguishable for any PPT adversary. 

    Combining all the partial hybrids, we can show that any PPT adversary cannot distinguish between {\bf Hybrid 0} and {\bf Hybrid 1}.

    {\bf Hybrid 2.} In this hybrid, we set $\alpha_i^1 = \vA_i\vs + T\ve_i + \vx_i - \vc_i$ for all $i \in [r]$. This hybrid is indistinguishable from {\bf Hybrid 1} by the pseudorandomness of $\MLWE$ samples.

    {\bf Hybrid 3.} In this hybrid, we can equivalently sample the distribution in the previous hybrid by first sampling uniformly random $\beta_{j,k}^i \ugets R_q^m$ for all $i \in [\ell], k \in [r]$ and $j >2$. Then we set
    $\beta_j^i = \sum_{k = 1}^r\bar{\vlambda}_k^i\beta_{j,k}^i$ for $j > 2$, and $\beta_i^1 =\sum_{j=1}^{t_i'}\sum_{k=1}^r\bar{\vlambda}^{i}_k(T\vf^{i}_{k,j}-\vA_k\vs^{i}_j) - \sum_{j=2}^{t_j'} \beta_j^{i}$ for all $i\in [\ell]$. By the same invariant, we can rewrite $\beta_i^1$ as
    \[
        \beta_i^1 = \sum_{k=1}^r\bar{\vlambda}^i_k(T\vf^{i}_k-\vA_k\vs)-\vc_i' - \sum_{j=2}^{t_j'} \beta_j^{i}
    \]
    where $\vf_k^i = \sum_{j=1}^{t_i'} \vf_{k,j}^{i},$ and $c_i' = \sum_{j=1}^r\sum_{k \text{ corrupt}} -\bar{\vlambda}^i_j\vA_js_k^i = -\sum_{i=1}^r \bar{\vlambda}^i_j\vc_i$.

    We can apply the same partial hybrid technique to show that {\bf Hybrid 2} and {\bf Hybrid 3} are indistinguishable.

    {\bf Hybrid 4.} In this hybrid, we set $\beta_i^1 = \sum_{k=1}^r\bar{\vlambda}^i_k(T\vf^{i}_k-\vA_k\vs)-\vc_i'$, and sample uniformly random $\beta_j^i \ugets R_q^m$. This hybrid is indistinguishable from {\bf Hybrid 3} by the pseudorandomness of $\MLWE$ samples and by the fact that $\bar{\vlambda}^i \in \bbZ_q^m \setminus \{\bf 0\}$ and $q$ is a prime.

    After {\bf Hybrid 4}, the view becomes
    \begin{align*}
    (\vA_1, \dots,\vA_{r},\alpha_{1}^1:=\vA_1\vs+T\ve_1+\vx_1-\vc_1,\alpha^{1}_{2},\ldots,\alpha^{1}_{t_1},\ldots,\alpha_1^r:=\vA_r\vs+T\ve_r+\vx_r-\vc_r,\alpha^{r}_{2},\ldots,\alpha^{r}_{t_r},\\\beta_{1}^1:=\sum_{i=1}^r\bar{\vlambda}_i^{1}(T\vf_i^{1}-\vA_i\vs + \vc_i), \beta_{2}^1,\ldots,\beta_{t_r'}^{1},\ldots, \beta_{1}^{\ell}:=\sum_{i=1}^r\bar{\vlambda}_i^{\ell}(T\vf_i^{\ell}-\vA_i\vs+\vc_{\ell}), \beta_{2}^{\ell},\ldots,\beta_{t_r'}^{\ell}
    )
    \end{align*}
    where $\alpha_{j}^i$ and $\beta_{k}^i$ are independent uniformly random variables over $R_q^m$ for all $j >2$ and $k > 2$.
    
    {\bf Hybrid 5.} In this hybrid, we sample uniformly random $\alpha_1^i$ for all $i \in [r]$, and set $\beta_1^k$ to $\sum_{i=1}^r\bar{\vlambda}_i^{k}(T\vf_i^{k}+T\ve_i -\alpha_1^i) + \langle\bar{\vlambda}^i,\vx\rangle$ for all $k \in [\ell]$. Because all $\alpha_1^i, \beta_1^k$ are independent of other random variables. So by Lemma~\ref{lem:key_no_dropout}, this hybrid is indistinguishable from {\bf Hybrid 4}.

    The last hybrid {\bf Hybrid 5} can be computed from the simulator's input. Therefore, the security claim follows.

\end{proof}

\end{document}